\documentclass[11pt,reqno]{amsart}
\usepackage{amsmath,amsfonts,amsthm,amssymb,amsxtra}
\usepackage[colorlinks,citecolor=red,hypertexnames=false]{hyperref}

\usepackage{color}
\usepackage{graphicx,subfigure}
\usepackage{float}
\usepackage{tikz}
\usetikzlibrary{calc}
  
\usepackage[shortlabels]{enumitem}

\usepackage{bbm}       
\usepackage{stmaryrd}  
\usepackage{nicematrix}
\usepackage{mathrsfs}  
\usepackage{marginnote}

\usepackage[margin=1.1in]{geometry}
\newtheorem{theorem}{Theorem}[section]

\newtheorem{conjecture}[theorem]{Conjecture}
\newtheorem{proposition}[theorem]{Proposition}
\newtheorem{lemma}[theorem]{Lemma}
\newtheorem{corollary}[theorem]{Corollary}

\theoremstyle{definition}

\newtheorem{assumption}{Assumption}

\theoremstyle{remark}
\newtheorem{remark}{Remark}[section]

\numberwithin{equation}{section}
\allowdisplaybreaks

\newcommand{\C}{\mathbb{C}}

\newcommand{\E}{\mathbb{E}}

\newcommand{\N}{\mathbb{N}}
\newcommand{\R}{\mathbb{R}}

\newcommand{\V}{\mathbb{V}}
\newcommand{\Z}{\mathbb{Z}}
\renewcommand{\P}{\mathbb{P}}

\DeclareMathOperator{\dist}{dist}

\DeclareMathOperator{\supp}{supp}

\newcommand{\ipc}[2]{ \langle #1 , #2  \rangle }

\newcommand{\braket}[3]{\left \langle #1 \, | #2 |\, #3 \right \rangle }

\newcommand{\bra}[1]{\left < #1 \right |}
\newcommand{\ket}[1]{\left | #1 \right >}

\newcommand{\one}{\mathbbm{1}}

\newcommand{\vol}{{\rm Vol}}

\newcommand{\vertiii}[1]{{\left\vert\kern-0.25ex\left\vert\kern-0.25ex\left\vert #1
    \right\vert\kern-0.25ex\right\vert\kern-0.25ex\right\vert}}

\renewcommand{\epsilon}{\varepsilon}
\renewcommand{\phi}{\varphi}

\newcommand{\eps}{\varepsilon}

\newcommand{\wt}{\widetilde}

 \renewcommand{\emptyset}{\varnothing}

\begin{document}

 \title[Anderson localization near the edge]{Edge localization and Lifshitz tails for graphs with Ahlfors regular volume growth}

\author[L. Shou, W. Wang, S. Zhang]{Laura Shou, Wei Wang, Shiwen Zhang}

\begin{abstract}
In this work, we study the Anderson model on graphs with Ahlfors $\alpha$‑regular volume growth. We show that, under mild regularity assumptions of the random distribution, Lifshitz‑tail type estimates near the bottom of the spectrum lead to exponential decay of fractional moments of the Green’s function and thus spectral and dynamical localization at low energies. This generalizes the result of \cite{aizenman01finitevol} from the lattice \(\Z^d\) to Ahlfors $\alpha$‑regular graphs. In addition, we establish Lifshitz tail estimates  for the integrated density of states, with the Lifshitz exponent determined by the
ratio of the volume growth rate and the random walk dimension of the underlying graphs, under certain assumptions on low lying eigenvalues of the Dirichlet and Neumann Laplacian on the graph. As an application, we verify all conditions on the Sierpinski gasket graph and obtain that, under mild regularity assumptions of the random distribution, for any fixed disorder, the Anderson model on the Sierpinski gasket graph has pure point spectrum and exhibits strong dynamical localization near the bottom of the spectrum.
 
\end{abstract}

  \maketitle
\tableofcontents

\section{Introduction}
In this work, we consider the Anderson model on  an (unweighted) graph $(\mathbb{V}, \mathcal{E})$, where the vertex set $\mathbb{V}$ is countably infinite. We write $x \sim y$ to indicate  that $\{x, y\} \in \mathcal{E}$, i.e., there is an edge between $x,y$.   The Anderson model is a family of random Schr\"odinger operators 
on ${\mathcal H}:=\ell^2(\V)$, defined as 
\begin{align}\label{eqn:AM}
    H_\omega f(x)= (-\Delta f+V_\omega)f(x)=\sum_{y\in \V:y\sim x}\big(f(x)- f(y)\big) + V_\omega(x)f(x),\  \ x \in \V,
\end{align}
where the first summation represents the (negative, combinatorial) graph Laplacian \(-\Delta\) on the graph, and $V_\omega=\{V_\omega(x)\}_{x\in \V}$ is a random potential.
The random potential $V_\omega=\{V_\omega(x)\}_{x\in \V}$ consists of independent, identically distributed (i.i.d.) random variables with a common distribution $P_0$ on $\R$, induced by a probability space $(\Omega,{\mathcal F},\P)$. We denote by ${\rm supp}P_0$ the (essential) support of the measure $P_0$, that is,
\begin{align}
    {\rm supp}P_0=\big\{\, E \in \R\, |\, P_0(E-\eps,E+\eps)>0 \ {\rm for\ all\ } \eps>0\big\}. 
\end{align}
Without loss of generality, we assume that $\inf {\rm supp}P_0=0$ and $p_0=P_0(0)<1$, i.e., the distribution is non-trivial and contains more than one point.

The graph $(\V,\mathcal E)$ is equipped with the natural graph metric, denoted by $d(\cdot,\cdot)$, which represents the length $n \in \mathbb{N}_0$ of the shortest path between two points. For $x \in \mathbb{V}$ and $r \ge 0$, define balls in $(\V,\mathcal E)$ with respect to $d$ as $B(x,r) = B(x,\lfloor r \rfloor)= \{ y \in \mathbb{V} : d(x,y) \le r \}=B(x,\lfloor r \rfloor),$ where $\lfloor \cdot \rfloor$ denotes the floor function. For any subset $S \subset \mathbb{V}$, we denote its cardinality by $|S| = \#\{ x : x \in S \}. $
Throughout this article, we work with graphs satisfying a polynomial volume growth condition: there exist $\alpha \ge 1$ and constants $c_1,c_2>0$ such that for any $x\in \V$ and $r\ge1$, 
        \begin{align}\label{eqn:vol-control}
       c_1r^{\alpha} \le  |B(x,r)| \le c_2 r^{\alpha}.
        \end{align}
This property is also referred to as ``Ahlfors $\alpha$-regular'' in \cite{barlow-alphabeta}, where the exponent $\alpha$ serves as the graph-theoretic analogue of the Hausdorff dimension. Note that \eqref{eqn:vol-control} also implies bounded geometry: letting $\deg(x) = \#\{ y : x \sim y \}$ be the degree of a vertex $x \in \V$, we have
    \begin{align}\label{eqn:bded-geo}
        \sup_{x\in \V}{\deg(x)}:=M<\infty.
    \end{align}  
When using \eqref{eqn:vol-control} or \eqref{eqn:bded-geo}, we may suppress the dependence on $c_1, c_2$, and $M$, focusing primarily on the parameter $\alpha$.

Random operators are known to exhibit spectral regimes where eigenvalues are dense and associated eigenfunctions are exponentially localized. There are two main approaches to studying localization in random operators: the \emph{multiscale analysis} (MSA), introduced by Fr\"ohlich and Spencer \cite{frohlich83}, and the \emph{fractional moments method} (FMM), introduced by Aizenman and Molchanov \cite{aizenman1993localization}. In this work, we adopt the latter approach, which is based on analyzing fractional moments of the Green's function. The FMM approach requires certain mild regularity assumptions on the probability distribution $P_0$. 
For later convenience, we define several assumptions in this section, which we may later reference in theorem statements.
\begin{assumption} \label{assume:holder}
 Given $0<\tau\le 1$,  the probability distribution $P_0$ is said to be (uniformly) $\tau$-H\"older continuous if for some $\kappa_\tau>0$,
 \begin{align}\label{eqn:tau-holder}
      \sup_{E\in \R} \, P_0\big([E-t,E+t]\big) \le \kappa_\tau\, t^{\tau}, \ \ t\ge 0. 
 \end{align}
\end{assumption}
\begin{remark}
Neither independence nor $\tau$-H\"older regularity is essential for the FMM method. Both conditions can be relaxed to more general assumptions; see, for example, \cite[Chapter 10]{aizenman2015random}.
\end{remark}

The FMM has been developed for various models and localization regimes, yielding increasingly stronger conclusions. In the present work, we focus on the localization of eigenfunctions of $H_\omega$ near the bottom of its spectrum.  It is well known that the density of states of random Schr\"odinger operators on $\ell^2(\mathbb{Z}^d)$ (or $L^2(\mathbb{R}^d)$) exhibits the so-called Lifshitz tail behavior near the bottom of the spectrum; see, e.g., \cite{kirsch2007invitation}. The close connection between Lifshitz tails and localization dates back to \cite{lif1965} in the physics literature. A rigorous proof of localization near the band edge using Lifshitz tails via the multiscale analysis approach was first given in \cite{martinelli84} and further developed in \cite{kirsch98,kirsch06}. In \cite{aizenman01finitevol}, the authors show that for random Schr\"odinger operators on $\ell^2(\mathbb{Z}^d)$, the fractional moments localization condition is satisfied in the regime where localization can be established via the (Lifshitz tail + MSA) approach. In this work, we aim to extend the approach (Lifshitz tail + FMM $\Rightarrow$ localization), developed in \cite{aizenman01finitevol} for the regular lattice $\mathbb{Z}^d$, to an ``$\alpha$-dimensional'' graph or lattice satisfying \eqref{eqn:vol-control}. Following \cite{aizenman01finitevol}, we assume the following fast power decay, which is slightly weaker than the usual (sub)exponential Lifshitz tail.

\begin{assumption}[Fluctuation Boundary: Fast Power-Decay Tails]\label{assume:lif}
Let $H^{B_R}$ denote the restriction of $H_\omega$ to a ball $B_R = B(x_0, R) \subset \V$ centered at $x_0 \in \V$ with radius $R$. Denote by $\inf \sigma(H^{B_R})$ the smallest eigenvalue of $H^{B_R}$. We say that $H_\omega$ exhibits a fast power-decay on $B_R$ if there exist constants $\delta, \alpha_0> 0$ such that
\begin{align}\label{eqn:power-decay}
        \P\big(\inf \sigma(H^{B_R})\le R^{-\delta}\big)\le   R^{-\alpha_0}.
\end{align}   
\end{assumption}

For $z \in \mathbb{C} \setminus \sigma(H_\omega)$, denote by $G(x,y;z)$ the Green's function, which is the kernel of the resolvent of $H_\omega$:
\begin{align}\label{eqn:Green}
    G(x,y;z):=\bra{x}(H_\omega-z)^{-1}\ket{y}=\ipc{\delta_x}{(H_\omega-z)^{-1}\delta_y}, 
\end{align}
with $G^{\Lambda}(x,y;z)$ defined analogously for any subset $\Lambda \subset \V$,  and where we used the usual inner product  $\langle f,g\rangle:=\sum_{x\in \V} {\bar f(x)}g(x)$ on \(\ell^2(\V;\C)\).
The subspaces $\ell^2(A)$ are defined accordingly for any finite subset $A\subset \V$. We often use the bra-ket convention to write $L(x,y)=\braket{x}{L}{y}=\langle \delta_x,L\delta_y\rangle$ for the kernel of an operator $L$.
The first result of this paper establishes that the fast power-decay condition (for a sufficiently large ball) implies decay of the Green's function.

\begin{theorem}\label{thm:FMbound}
Let $H_\omega$ be the Anderson model given in \eqref{eqn:AM} on a graph $(\V, \mathcal{E})$ satisfying \eqref{eqn:vol-control} for some $\alpha > 0$. Assume that:
\begin{enumerate}
\item Assumption~\ref{assume:holder} holds with $\tau \in (0,1)$ and $\kappa_\tau > 0$.
\item For some $\delta \in (0,1)$, $\alpha_0 > 3\alpha$, and $R$ sufficiently large (depending on $\delta, \alpha_0, \tau$), such that for any $x_0 \in \V$,  Assumption~\ref{assume:lif} holds on $B(x_0, R)$.
\end{enumerate}

Then:
\begin{enumerate}
\item There exist constants $\mu, C, E_0 > 0$ and $s \in (0, \tau/2)$ such that
\begin{align}\label{eqn:FM-bound}
\mathbb{E}\Big( \big| G(x,y; E + i\varepsilon) \big|^s \Big) \le C e^{-\mu\, d(x,y)}
\end{align}
for all $x,y \in \V$, $E \le E_0$, and $\varepsilon \neq 0$.

   \item If, in addition, $\tau = 1$ and $P_0$ has compact support, then the spectrum of $H_\omega$ within $I = [0, E_0]$ is either empty or pure point with strong dynamical localization on $I$. That is, for some $D_I > 0$ and all $R > 0$,  
\begin{align}\label{eqn:DAL}  
    \sup_{x \in \V} \sum_{\substack{y \in \V \\ d(x,y) \ge R}} \mathbb{E}\Big( \sup_{t \in \mathbb{R}} \big| \ipc{\delta_x}{P_I(H_\omega)e^{- i t H_\omega}\delta_y} \big|^2 \Big) \le D_I e^{-\mu R},  
\end{align}  
where $P_I(H_\omega)$ is the spectral projection of $H_\omega$ onto the interval $I$.  
    \end{enumerate}
\end{theorem}

\begin{remark}
 It is a well-established result that for a compactly supported random distribution, fractional moments of the Green's function provide an upper bound on the eigenfunction correlator of the random operator; see \cite[Theorem 7.7]{aizenman2015random}. Consequently, the exponential decay of the Green's function in \eqref{eqn:FM-bound} implies exponential decay of the eigenfunction correlator, which in turn leads to pure point spectrum and strong dynamical localization as expressed in \eqref{eqn:DAL}; see \cite[\S7.1]{aizenman2015random}. Our primary focus will be on establishing \eqref{eqn:FM-bound}. 
  One can verify that \eqref{eqn:DAL} implies that for any $q > 0$, the $q$-th moment of the wave packet for any initial data $\delta_x$ is uniformly bounded in time:
    \begin{align}\label{eqn:DAL-q}
   \sup_{t\in \R}\,  \sup_{x\in \V} \,  \sum_{y\in \V}\E\Big( d(x,y)^q\big|\ipc{\delta_x}{P_I(H_\omega)e^{- i t H_\omega}\delta_y}\big|^2\Big)\, <\infty  , 
    \end{align}
 which is often used as an alternative definition of dynamical localization.
\end{remark}

\begin{remark}
  Given a random operator $T + \lambda V_\omega$ on a graph $(\V,\mathcal{E})$, where $T$ is a bounded, self-adjoint (non-random) hopping term and $\lambda > 0$ is a coupling constant, it is a classical result that  for sufficiently large disorder $\lambda > \lambda_T$, the random operator exhibits strong exponential dynamical localization throughout its spectrum, provided the random distribution of $V_\omega$ is sufficiently regular; see \cite[Theorem 10.2]{aizenman2015random}. Such localization at large disorder does not require any regularity or dimensionality assumptions on the underlying graph and, in particular, applies to the random Schr\"odinger operator \eqref{eqn:AM} considered in this work. Our main interest, however, lies in the case where the disorder strength $\lambda$ is mild or small.
On the Euclidean lattice $\mathbb{Z}^d$ (where $\alpha = d$), for the random operator $T + \lambda V_\omega$ with any fixed disorder strength $\lambda$, it was proved in \cite[Theorem 4.3]{aizenman01finitevol} using the FMM method that a Lifshitz-tail type estimate implies localization near the spectral edge; see also \cite[Corollary 11.6]{aizenman2015random}. Theorem~\ref{thm:FMbound} extends this fundamental result from $\mathbb{Z}^d$ to any ``Ahlfors $\alpha$-regular'' graph\footnote{For the $\mathbb{Z}^d$ lattice, Assumption~\ref{assume:lif} can be weakened to $\alpha_0 > 3(d-1)$ due to the boundary geometry $B_R \setminus B_{R-1} \sim R^{d-1}$ for any ball or cube of radius or side length $\sim R$.} with non-integer $\alpha \ge 1$.
\end{remark}

The second main result of this paper provides a characterization of Lifshitz-tail estimates for random operators of the form $H_\omega=-\Delta+V_\omega$. A sufficient condition for obtaining such estimates involves the Neumann eigenvalue of the free Laplacian on a ball. Let $\Delta^{B_r,N}$ denote the Neumann Laplacian (i.e., the subgraph Laplacian) on a ball $B_r = B(x_0, r) \subset \V$ for some $x_0 \in \V$ and $r > 0$:
\begin{align}\label{eqn:N-lap} 
  \Delta^{B_r,N}f(x)=  \sum_{ y\in B_r :y\sim x}\big(f(y)- f(x)\big), \ \ x\in  B_r.
\end{align}
Denote by $E_1 = E_1(-\Delta^{B_r,N})$ the first (smallest) non-zero eigenvalue of $-\Delta^{B_r,N}$.
\begin{assumption}[Neumann Laplacian Eigenvalue Lower Bound]\label{assume:N-ev}
We say that $E_1$ satisfies a $\beta$-power lower bound on some ball $B(x_0, r)\subset \V$ if there exist constants $\beta, c_0 > 0$ such that 
\begin{align}\label{eqn:N-ev-lower}
    E_1(- \Delta^{B_r,N}) \ge \frac{c_0}{r^\beta}. 
\end{align}
\end{assumption}

\begin{theorem}\label{thm:lif}
Let $H_\omega$ be the Anderson model given in \eqref{eqn:AM} on a graph $(\V,\mathcal{E})$ satisfying \eqref{eqn:vol-control} for some $\alpha > 0$. Assume that for given $\delta \in (0,1)$ and $\beta \in [2,\alpha+1]$, there exists $c_0 > 0$ such that for sufficiently large $r$ and any $x_0 \in \V$, condition \eqref{eqn:N-ev-lower} holds on $B(x_0,r)$ with constants $c_0, \beta$. Then there exist constants $\eta > 0$, $C > 0$, and $C_0$ (independent of $r$) such that, setting $R = C_0 r^{\frac{\beta}{\delta}}$, for any $x \in \V$ and $B_R = B(x,R)$,
\begin{align}\label{eqn:exp-lif-upper}
        \P\big(\inf \sigma(H^{B_R})\le R^{-\delta}\big)\le C   \exp(-\eta R^{\delta\, \frac{\alpha}{\beta} }).
\end{align}

  As a consequence, the fast power-decay condition \eqref{eqn:power-decay} holds for any $\alpha_0 > 0$, provided $R = C_0 r^{\frac{\beta}{\delta}}$ is sufficiently large.
\end{theorem}

This result can be interpreted as a finite-volume analogue of the Lifshitz-tail (upper) estimate, which suffices to establish localization. If Assumption~\ref{assume:N-ev} holds for infinitely many $r \nearrow \infty$, the argument used to prove \eqref{eqn:exp-lif-upper} can be adapted to derive Lifshitz upper bounds for the infinite-volume density of states (assuming it exists). Denote by $\mathcal N(E;X)$ the eigenvalue counting function below energy $E$ for an operator $X$:
\begin{align}\label{eqn:ev-count}
    \mathcal N(E;X)=\#\big\{\ {\rm eigenvalues}\ E'\ {\rm of}\ X \ {\rm such\ that\ }\ E'\le E\ \big\}.
\end{align}

\begin{theorem}\label{thm:infi-lif-upper}
Let $H_\omega$ be the Anderson model given in \eqref{eqn:AM} on a graph $(\V,\mathcal{E})$ satisfying \eqref{eqn:vol-control} for some $\alpha > 0$. Suppose that for $\beta \in [2,\alpha+1]$, there exist constants $c_0, c_3 > 0$ such that there is an increasing sequence $\{r_i\}_{i \in \mathbb{N}}$ with $r_i \nearrow \infty$ and $r_{i+1} \le c_3 r_i$ for all $i \in \mathbb{N}$, and for any $x_0 \in \V$, 
Assumption~\ref{assume:N-ev} holds on $B(x_0, r_i)$ with constants $c_0, \beta$.
Then there exist constants $C_1, C_2, E_0>0$ such that for any $x \in \V$ and any $E \le E_0$,
\begin{align}\label{eqn:infi-lif-up}
    \varlimsup_{R\to \infty}  \frac{1}{|B(x,R)|}\E \mathcal N(E;H^{B(x,R)})  \le C_1 \exp(-C_2E^{-\frac{\alpha}{\beta}}).
\end{align}
\end{theorem}

Conversely, a comparable power upper bound on the Dirichlet eigenvalue yields Lifshitz lower bounds for the infinite-volume density of states (assuming it exists).
Let $\Delta^{B_r}$ denote the Dirichlet Laplacian on the ball $B_r = B(x_0, r) \subset \V$ for some $x_0 \in \V$ and $r > 0$:
\begin{align}\label{eqn:D-lap}
\Delta^{B_r}f(x)=  \sum_{ y\in \V :y\sim x}\big(f(y)- f(x)\big), \ \ x\in  B_r.
\end{align}
Let $E_1 = E_1(-\Delta^{B_r})$ denote the first (smallest) eigenvalue of $-\Delta^{B_r}$.
\begin{assumption}[Dirichlet Laplacian Eigenvalue Upper Bound]\label{assume:D-ev}
We say that $E_1$ satisfies a $\beta$-power upper bound on a ball $B(x_0, r)$ if there exist constants $\beta, c_0' > 0$ such that
\begin{align}\label{eqn:D-ev-upper}
E_1(- \Delta^{B_r}) \le \frac{c_0'}{r^\beta}.
\end{align}
\end{assumption}

\begin{theorem}\label{thm:infi-lif-lower}
Let $H_\omega$ be the Anderson model given in \eqref{eqn:AM} on a graph $(\V,\mathcal{E})$ satisfying \eqref{eqn:vol-control} for some $\alpha > 0$.
Suppose that for $\beta \in [2,\alpha+1]$, there exist constants $c_0', c_3' > 0$ and an increasing sequence ${r_i}{i\in\N}$ with $r_i\nearrow \infty$ and $r{i+1}\le c_3'r_i$ for all $i \in \mathbb{N}$, and that for any $x_0 \in \V$,  Assumption~\ref{assume:D-ev} holds on $B(x_0, r_i)$ with constants $c_0', \beta$.
Assume in addition that the common distribution $P_0$ satisfies  \begin{align}\label{eqn:V-Lif-ass}
P_0([0,\eps])\ge C\eps^\kappa,
\end{align}
for some $C, \kappa > 0$ and all sufficiently small $\varepsilon > 0$.
Then there exist constants $C_1', C_2', E_0'>0$ 
such that for any $x \in \V$ and any $E \le E_0'$,
\begin{align}\label{eqn:infi-lif-lower}
\varliminf_{R\to \infty}  \frac{1}{|B(x,R)|}\E \mathcal N(E;H^{B(x,R)})  \ge C_1' \exp(-C_2' |\log E| E^{-\frac{\alpha}{\beta}})
\end{align}
\end{theorem}

\begin{remark}
Combining \eqref{eqn:infi-lif-up} and \eqref{eqn:infi-lif-lower} yields the following asymptotic characterization:
     \begin{align}\label{eqn:lif-log-log}
     \begin{aligned}
     \lim_{E\searrow 0}    \varliminf_{R\to \infty} \frac{1}{\log E}\log\left|\log\left(\frac{\E \mathcal N(E;H^{B(x,R)}) }{|B(x,R)|}  \right)\right|    &=  -\frac{\alpha}{\beta},\\
          \lim_{E\searrow 0}    \varlimsup_{R\to \infty} \frac{1}{\log E}\log\left|\log\left(\frac{\E \mathcal N(E;H^{B(x,R)}) }{|B(x,R)|}  \right)\right|    &=  -\frac{\alpha}{\beta}.
        \end{aligned}
     \end{align}
Note that the lower tail bound in \eqref{eqn:V-Lif-ass} is necessary for the lower-bound exponent in \eqref{eqn:lif-log-log}. If the tail of the underlying random distribution is ``thinner'' than \(O(\varepsilon^\kappa)\), one expects a smaller Lifshitz exponent than \(-\alpha/\beta\); see \cite{pastur1992book}.
 \end{remark}
\begin{remark}
We are interested in whether the parameters $\beta$ in Assumptions~\ref{assume:N-ev} and \ref{assume:D-ev} always coincide or not.
For any \(d\in \N_+\), on the \(d\)-dimensional Euclidean lattice \(\Z^d\), Assumptions~\ref{assume:N-ev} and \ref{assume:D-ev} hold with the same \(\beta=2\) for every ball \(B(x,r)\subset\Z^d\) by direct computation. 
The $\beta$ in Assumption~\ref{assume:D-ev} is often referred to as the random-walk dimension and is linked to the second parameter of the \emph{Heat Kernel Bound} $\mathrm{HK}(\alpha,\beta)$, a property of the free Laplacian on the corresponding space. The first parameter \(\alpha\) is the ``Ahlfors-volume'' parameter in \eqref{eqn:vol-control}. These two parameters admit many equivalent characterizations via properties of random walks and heat kernels on graphs; see the thorough discussion in \cite{barlow2017random}. We are not aware of any result showing that the \emph{Heat Kernel Bound} $\mathrm{HK}(\alpha,\beta)$ implies the Neumann eigenvalue lower bound in \eqref{eqn:N-ev-lower}. Consequently, it is unclear whether the \emph{Heat Kernel Bound} $\mathrm{HK}(\alpha,\beta)$ suffices to prove Theorems~\ref{thm:FMbound}, \ref{thm:infi-lif-upper}, and \ref{thm:infi-lif-lower}.
\end{remark}

A notable example of graphs with explicit parameters \(\alpha,\beta\) in which all these assumptions hold is the Sierpinski gasket graph. The spectrum and Lifshitz-tail singularity of the integrated density of states for the Anderson model on this graph have been studied by the authors in \cite{shou2024spectrum}. The estimate \eqref{eqn:lif-log-log} extends the result for the Sierpinski gasket graph to more general Ahlfors \(\alpha\)-regular graphs. We briefly define the Sierpinski gasket graph below to state our result on Anderson localization. See \cite{shou2024spectrum} and references therein for additional background on the Sierpinski gasket graph and related discussions of Lifshitz-tail estimates for random operators.  
 
Let $T_0\subseteq \R^2$ denote the unit equilateral triangle with vertex set  $\V_0=\{(0,0),(1,0),(1/2,\sqrt 3/2)\}=\{ a_1,a_2,a_3\}$.  
Define $T_n$ recursively by
\begin{align}\label{eqn:T-rec}
  T_{n+1}=T_n\cup \big(T_n+2^na_2\big) \cup \big(T_n+2^na_3\big),
  \ n\ge 0. 
\end{align}
Each $T_n$ consists of $3^n$ translations of the unit triangle $T_0$. 
Let $\V_n$ be the set of all vertices of the triangles in $T_n$. The edge set $\mathcal E_n=\{(x,y):x,y\in \V_n\}$ is defined by declaring $(x,y)\in \mathcal E_n$ if and only if there exists a triangle $T$ of side length 1 in $T_n$ with $x,y\in T$. The Sierpinski gasket graph $\mathcal{SG} = (\V_{\rm SG}, \mathcal{E}_{\rm SG})$ is defined as
\begin{align}\label{eqn:SG-def}
    \V_{\rm SG}=\bigcup_{n\ge 0}(\V_n\, \cup \,  \V_n'), \ \ \mathcal E_{\rm SG}=\bigcup_{n\ge 0}(\mathcal E_n\, \cup \mathcal E_n'),  
\end{align}
where $\V_n'$ and $\mathcal E_n'$ are the reflections of $\V_n$ and $\mathcal E_n$ across the $y$-axis, respectively; see Figure~\ref{fig:sg}. 
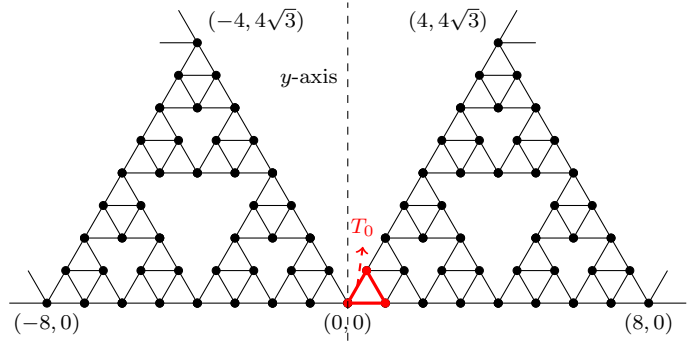
\begin{figure}[!ht]
\centering
\def\points{{0,0},{1,0},{.5,.866025},{1.5,.866025},{1,0},{2,0},{1,1.73205},{0,0}}

\begin{tikzpicture}[scale=.5] 
\foreach \shift in {(0,0),(2,0),(4,0),(6,0),(1,1.73205),(5,1.73205),(2,3.4641),(4,3.4641),(3,5.19615)}{
    \foreach \c[remember=\c as \clast (initially {0,0})] in \points {
        \draw[shift=\shift cm] (\clast) -- (\c);
        \filldraw[shift=\shift cm, black] (\c) circle (3pt);
    }
}

\begin{scope}[xscale=-1]
\foreach \shift in {(0,0),(2,0),(4,0),(6,0),(1,1.73205),(5,1.73205),(2,3.4641),(4,3.4641),(3,5.19615)}{
    \foreach \c[remember=\c as \clast (initially {0,0})] in \points {
        \draw[shift=\shift cm] (\clast) -- (\c);
        \filldraw[shift=\shift cm, black] (\c) circle (3pt);
    }
}
\end{scope}

\draw (4,6.9282) -- ++(60:1cm);
\draw (4,6.9282) -- ++(0:1cm);
\draw (8,0)--(9,0);
\draw (8,0) -- ++(60:1cm);

\begin{scope}[xscale=-1]
\draw (4,6.9282) -- ++(60:1cm);
\draw (4,6.9282) -- ++(0:1cm);
\draw (8,0)--(9,0);
\draw (8,0) -- ++(60:1cm);
\end{scope}
\node[below, font=\scriptsize] at (0,0) {$(0,0)$};
\node[below,font=\scriptsize] at (8,0) {$(8,0)$};
\node[below,font=\scriptsize] at (-8,0) {$(-8,0)$};
\node[above left,font=\scriptsize] at (4,6.9282) {$(4,4\sqrt{3})$};
\node[above right,font=\scriptsize] at (-4,6.9282) {$(-4,4\sqrt{3})$};

\draw[red, line width=1.2pt]
  (0,0) -- (1,0) -- (.5,.866025) -- cycle;
\foreach \pt in {(0,0),(1,0),(.5,.866025)}{
  \filldraw[red] \pt circle (3pt);
}
\draw[red, dashed, ->, line width=0.8pt]
  (0.25,0.4330125) -- (0.4,1.5)
  node[above, red, font=\scriptsize] {$T_0$};

\draw[dashed] (0,-1) -- (0,8); 
\node[left,font=\scriptsize] at (0,6) {$y$-axis};
\end{tikzpicture}
\caption{The unit triangle $T_0$ is located next to the origin, with vertices $\V_0=\{(0,0),(1,0),(1/2,\sqrt3/2)\}$.
The right hand side of the $y$-axis is the 3rd step of construction $T_3$ and the left hand side is its reflective mirror with respect to the $y$-axis. The picture contains $2\times 27$ many unit triangles $T$, which are all translations of $T_0$.  The dots form the vertex set $\V_3'\cup\V_3$. The edge set $\mathcal E_3'\cup\mathcal E_3$ consists of all edges of length 1 of the unit triangles.}
\label{fig:sg}
\end{figure}

The Sierpinski gasket graph $\mathcal{SG}$ is a classical self-similar graph satisfying \eqref{eqn:vol-control} with $\alpha=\log3/\log2\approx 1.58\ldots$, namely a fractal graph of that dimension; see \cite{BarlowPerkins} and \cite[\S2.9]{barlow2017random}. In our previous work \cite{shou2024spectrum}, we proved that Assumption~\ref{assume:N-ev} holds with $\beta=\log5/\log3$ (which is the same value of $\beta$ as in Assumption~\ref{assume:D-ev}), yielding a Lifshitz-tail estimate similar to \eqref{eqn:exp-lif-upper}. As a consequence of Theorems~\ref{thm:FMbound}, \ref{thm:lif}, and \cite[Theorems 1.1, 1.2]{shou2024spectrum}, we obtain:

\begin{corollary}
For the Anderson model $H_\omega$ on the Sierpinski gasket graph $\mathcal{SG}$, the following hold:
\begin{enumerate}[(i)]
\item For any \(\delta\in (0,1)\), there exist constants \(\eta>0\), \(C>0\), and a sequence of balls with radii \(R\), \(R\to\infty\) such that the Lifshitz-tail estimate \eqref{eqn:exp-lif-upper} holds with $\alpha = {\log 3}/{\log 2}$ and $\beta = {\log 5}/{\log 3}$. 
\item Suppose Assumption~\ref{assume:holder} holds with $\tau = 1$ and $P_0$ has compact support. Then there exists $E_0 > 0$ such that, if $\sigma(H_\omega) \cap [0, E_0] \neq \emptyset$, the operator $H_\omega$ has pure point spectrum and exhibits dynamical localization as in \eqref{eqn:DAL} on $[0, E_0]$. 

If, in addition, ${\rm supp}\, P_0$ contains an interval $[0, E_1]$ at the bottom of the spectrum, then $[0, E_1] \subset \sigma(H_\omega)$; accordingly, the above localization conclusion applies on $[0,\min(E_0, E_1)]$.
 \end{enumerate}

\end{corollary}

\begin{remark}
 On the Sierpinski gasket graph, the natural geometric unit is not a general graph ball but the finite triangle \(\V_n\) of ``side length'' \(2^n\), which is in fact a ``semiball'' of radius \(R=2^n\). It is more straightforward to verify Assumption~\ref{assume:lif},\ref{assume:N-ev} in terms of the sequence \(2^n\). See more details in \cite[Theorem 1.2]{shou2024spectrum}. 
 Also note that, to apply Theorem~\ref{thm:FMbound} and \ref{thm:lif} to prove localization results in Part (ii), we only need \eqref{eqn:exp-lif-upper} to hold for sufficiently large \(R\). 

In \cite[Theorem 1.1]{shou2024spectrum}, we have shown that for the Anderson model $-\Delta+V_\omega$  on the Sierpinski gasket graph \(\mathcal SG\), almost surely,  
    \begin{align}\label{eqn:AM-as-spectrum1}
 \sigma(-\Delta)+ \supp P_0      \subseteq  \sigma(-\Delta+V_\omega)\subseteq \sigma(-\Delta)+[\inf V_\omega,\sup V_\omega] . 
    \end{align}
Equality holds if the support of the random distribution has no gaps, e.g., a uniform distribution. In Part (ii), we do not require this stronger gapless-support assumption. This is because \(\inf \sigma(-\Delta)=0\), and we only use the left-hand inclusion 
\begin{align}
    [0,E_1]\subset 0+{\rm supp}\, P_0 \subset \sigma(-\Delta)+ \supp P_0      \subseteq  \sigma(-\Delta+V_\omega).
\end{align}
This ensures that the bottom of the spectrum contains an interval. 

For a singular distribution like the Bernoulli case, numerical evidence (see \cite[Figure 2]{shou2024spectrum}) suggests that the bottom also contains an interval, but we do not have a proof yet. In addition, the Bernoulli distribution violates the regularity assumption \eqref{eqn:tau-holder}, which precludes applying the FMM approach to obtain Anderson localization. The Anderson–Bernoulli localization is more challenging even in Euclidean spaces. 
Most existing results rely on multiscale analysis approach; see, e.g., the one-dimensional case in \cite{carmona,damanik02duke}, the two-dimensional case (\(\R^2\)) by Bourgain–Kenig \cite{bourken05}, and subsequent developments in the discrete setting in \cite{dingsmart20,lizhang22,li22}.
\end{remark}

There is extensive physics literature addressing the dependence of the metal-insulator transition (MIT) on dimensionality for the Anderson model, 
viewed as an extension of the conjecture for Euclidean spaces; see e.g. \cite{abrahams1979scaling,wegner1989four,lerner1994spectral,schreib1996}. 
A general expectation from the physics literature connects the presence (or absence) of weak localization to recurrence (or transience) of a classical random-walking particle on the underlying space \cite{AltshulerAranov1985,lee1985disordered,curtis2025absence,chen2024anderson,shou2025geometric}.
Graphs satisfying ${\rm HK}(\alpha,\beta)$ are transient if $\alpha>\beta$ and recurrent if $\alpha\le\beta$ (see e.g. \cite{barlow2017random}). Hence, if $\alpha\le\beta$, one expects the Anderson model to be fully localized for any disorder strength, and if $\alpha>\beta$, one expects full localization at large disorder and an MIT for small disorder. Mathematical progress has been incomplete, with only partial results available for the Euclidean lattice $\mathbb{Z}^d$ \cite{aizenman01finitevol,klopp2002weak,elgart2009lifshitz}. (See also results for the Bethe lattice \cite{aizenman2013resonant}.) In particular, there is little work in the mathematical literature on these questions for Ahlfors \(\alpha\)-regular graphs with non-integer dimension \(\alpha\).

The current work is our first attempt to investigate localization on these types of graphs. 
As an aside, we note that, while we consider the example of the Sierpinski gasket graph, constructed as a planar subset of \(\R^2\), there are also higher-dimensional analogues, called \(d\)-dimensional Sierpinski simplex graphs, such as the Sierpinski tetrahedron graph in $\mathbb{R}^3$ (Figure~\ref{fig:SP}). 
The corresponding  volume dimension is $\alpha_d=\frac{\log (d+1)}{\log 2}$, and the random-walk dimension is $\beta_d=\frac{\log(d+3)}{\log 2}$ for $d\ge 2$ (see \cite[\S10]{BarlowPerkins}); see also \cite{rammal84,alexander1982density,ben2000diffusion}. Recall that the Euclidean lattice satisfies \({\rm HK}(d,2)\), where the walk dimension remains constant at 2 as the volume dimension \(d\to\infty\). Full localization for the Euclidean lattice \(\Z^d\) is expected only for \(d\le 2\); this has been proved for \(\Z^1\) and remains open for \(\Z^2\). In contrast, for the \(d\)-dimensional Sierpinski simplex graph, we note that \(\alpha_d<\beta_d\) for all \(d\ge 2\), even though \(\alpha_d/\beta_d\nearrow 1^-\) as \(d\to \infty\).  Thus we are curious about the following conjecture.
\begin{conjecture}\label{conj:sg-loc}
      Let $H_\omega=-\Delta+  V_\omega$ the Anderson model  on the \(d\)-dimensional Sierpinski simplex graph \((\V,\mathcal E)\)  for \(d\ge 2\). For any random distribution of $V_\omega$, $H_\omega$ exhibits Anderson localization throughout its spectrum; that is, it has pure point spectrum with exponentially localized eigenfunctions forming a complete basis in $\ell^2(\mathbb{V})$.
\end{conjecture}

\newcommand{\DrawTetraEdges}[4]{%
  \draw[edge, ultra thin] (#1) -- (#2) -- (#3) -- cycle;
 \draw[edge, ultra thin] (#1) -- (#4);
\draw[edge, ultra thin] (#2) -- (#4);
\draw[edge, ultra thin] (#3) -- (#4);
  \fill (#1) circle[radius=2pt];
  \fill (#2) circle[radius=2pt];
  \fill (#3) circle[radius=2pt];
  \fill (#4) circle[radius=2pt];
}

\newcommand{\SierTetra}[5]{%
  \pgfmathtruncatemacro{\n}{#1}%
  \ifnum\n=0
    \DrawTetraEdges{#2}{#3}{#4}{#5}%
  \else
    \pgfmathtruncatemacro{\m}{\n-1}%
    \SierTetra{\m}
      {#2}
      {$ (#2)!0.5!(#3) $}
      {$ (#2)!0.5!(#4) $}
      {$ (#2)!0.5!(#5) $}
    \SierTetra{\m}
      {#3}
      {$ (#3)!0.5!(#2) $}
      {$ (#3)!0.5!(#4) $}
      {$ (#3)!0.5!(#5) $}
    \SierTetra{\m}
      {#4}
      {$ (#4)!0.5!(#2) $}
      {$ (#4)!0.5!(#3) $}
      {$ (#4)!0.5!(#5) $}
    \SierTetra{\m}
      {#5}
      {$ (#5)!0.5!(#2) $}
      {$ (#5)!0.5!(#3) $}
      {$ (#5)!0.5!(#4) $}
  \fi
}

\begin{figure}[H]
\centering
\begin{tikzpicture}[scale=0.7, line cap=round, line join=round]
  \tikzset{edge/.style={black, line width=0.6pt, opacity=0.95}}

  \begin{scope}[shift={(0,0)}]
    \coordinate (L1) at (0,0);
    \coordinate (L2) at (1.9,0.3);
    \coordinate (L3) at (1.6,-0.9);
    \coordinate (L4) at (1,1.5);
    \DrawTetraEdges{L1}{L2}{L3}{L4}
  \end{scope}

  \begin{scope}[shift={(7,0)}]
    \coordinate (A) at (-2,0);   
    \coordinate (B) at ( 2,0.5);   
    \coordinate (C) at ( 1.3,-2);
    \coordinate (D) at ( 0.2, 3.2);

    \SierTetra{1}{A}{B}{C}{D}
  \end{scope}

\end{tikzpicture}
 \caption{
    The first two generations of the Sierpinski tetrahedron  (Sierpinski pyramid) graph. The infinite graph  has volume dimension $\alpha_3=\log4/\log2$ and walk dimension $\beta_3=\log 6/\log2$.}
    \label{fig:SP}
\end{figure}
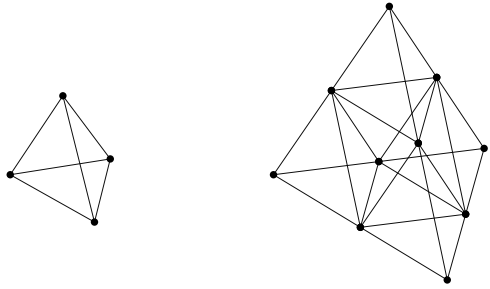

The remainder of this article is organized as follows. 
 In Section~\ref{sec:FM}, we prove Theorem~\ref{thm:FMbound}, namely the fractional-moment bound under Assumptions~\ref{assume:holder} and \ref{assume:lif}.
  In Section~\ref{sec:lif}, we establish the Lifshitz-tail estimates of Theorems~\ref{thm:lif}, \ref{thm:infi-lif-upper}, and \ref{thm:infi-lif-lower}, proving first the upper bounds \eqref{eqn:exp-lif-upper} and \eqref{eqn:infi-lif-up} via Neumann bracketing, and then the lower bound \eqref{eqn:infi-lif-lower} via (modified) Dirichlet bracketing.

Throughout the paper, constants such as $C$, $c$, and $c_i$ may change from line to line. 
We use the notation $X\lesssim Y$ to mean $X\le cY$, and $X\gtrsim Y$ to mean $X\ge cY$, for some
 constant $c$ depending only on the graph $(\V,\mathcal E)$. If $X\lesssim Y\lesssim X$, we may also write $X\approx Y$.


\section{Lifshitz Tails Imply Fractional Moment Bounds for the Green's Function}\label{sec:FM}

 In this section we show that the power-law bound \eqref{eqn:power-decay} implies the fractional-moment bound \eqref{eqn:FM-bound}. This extends the approach (Lifshitz tail + FMM $\Rightarrow$ localization) developed in \cite[Theorem 4.3]{aizenman01finitevol} for random operators acting on \(\ell^2(\Z^d)\). For a self-contained introduction to Anderson localization, see \cite{gunter2011intro}, especially Section~7 therein for an overview of this approach on \(\Z^d\). 

\subsection{Lifshitz Tail and Initial Estimate}
We first show that the Lifshitz-tail estimate \eqref{eqn:power-decay} establishes an initial-scale bound for the Green's function.

\begin{lemma}\label{lem:initial}
Let $H_\omega$ denote the Anderson model given in \eqref{eqn:AM}, which satisfies \eqref{eqn:vol-control} and \eqref{eqn:tau-holder} with $\alpha \ge 1$, $\tau \in (0,1)$, and $\kappa_\tau > 0$. Assume Assumption~\ref{assume:lif} holds for some $\delta \in (0,1)$ and $\alpha_0 > 3\alpha$.
Then there exist $s \in (0,\tau)$, $\alpha_1 \in (3\alpha,\alpha_0)$, and $C > 0$
such that, for $R$ in Assumption~\ref{assume:lif} taken sufficiently large, for all $x_0 \in \V$, $x,y \in B_R = B(x_0,R)$ with $d(x,y) \ge R/2$, all $E < \frac{1}{2}R^{-\delta}$, and $\eps > 0$, we have
    \begin{align}\label{eqn:initial-Gs}
     \E\Big(  |G^{B_R}(x,y;E+i\eps)|^s\Big) \le  CR^{-\alpha_1}.  
\end{align}
\end{lemma}

The key ingredients in the proof of Lemma~\ref{lem:initial} are a priori bounds on fractional moments (Lemma~\ref{lem:priori-G}) and Combes--Thomas estimates for energies away from the spectrum (Lemma~\ref{lem:CT}). These two lemmas are well-known results that hold in more general settings; we include the versions applicable to our model in Appendix~\ref{app:pB-CT}. We use these two lemmas to give
\begin{proof}[Proof of Lemma~\ref{lem:initial}]
Let $\Omega_B$ denote the event in the probability estimate \eqref{eqn:power-decay}, and let $\Omega_G$ be its complement, i.e.,
      \begin{align}
        \Omega_B=\big\{\,  \omega\in \Omega:\,  \inf \sigma(H^{B_R})\le R^{-\delta} \big\},\  \ \Omega_G=\Omega_B^c. 
    \end{align}
    Then for $z=E+i \eps$, 
    \begin{align}\label{eqn:good-bad-Gs}
      \E\big(  |G^{B_R}(x,y;z)|^s\big) = \E\Big(  |G^{B_R}(x,y;z)|^s\, \chi_{\Omega_B}\Big)+\E\Big(  |G^{B_R}(x,y;z)|^s\,\chi_{\Omega_G}\Big), 
    \end{align}
   where $\chi_S$ denotes the characteristic function of a set $S \subset \Omega$.
    
  On the ``bad set'' $\Omega_B$, the energy $E$ may resonate with the eigenvalues of $H^{B_R}$. We use the a priori bound Lemma~\ref{lem:priori-G} together with the probability estimate \eqref{eqn:power-decay} to control the fractional moments. More precisely, for any $s \in (0,\tau) \subset (0,1]$, let $p = \frac{1}{2}\big(1 + \frac{\tau}{s}\big) > 1$ so that $s' = sp < 1$, and let $q = \frac{p}{p-1}$ be conjugate to $p$ with $1/p + 1/q = 1$. Applying H\"older’s inequality with $(p,q)$ to the first term in \eqref{eqn:good-bad-Gs} gives for some constant \(C=C(s',p)>0\), 
  \begin{align}
     \E\Big(  |G^{B_R}(x,y;z)|^s\, \chi_{\Omega_B}\Big) \le & \Big| \E\big(  |G^{B_R}(x,y;z)|^{sp}\, \big)\Big|^{1/p}\, \Big| \E\big(    \chi_{\Omega_B}^q \big)\Big|^{1/q}  \notag \\
    = &   \Big| \E\big(  |G^{B_R}(x,y;z)|^{s'}\, \big)\Big|^{1/p}\,  \P\big(     \Omega_B  \big) ^{1/q} \notag \\
 \le & C\,   R^{-\frac{\alpha_0}{q}}, \label{eqn:bad-est0}   
  \end{align}
using the a priori bound \eqref{eqn:priori-xy} and \eqref{eqn:power-decay}.

Take $s = \tau \frac{\alpha_0 - 3\alpha}{3\alpha_0 + 3\alpha} \in (0,\tau)$; then $p = \frac{1}{2}\big(1 + \frac{\tau}{s}\big) = \frac{2\alpha_0}{\alpha_0 - 3\alpha} > 1$ and $\alpha_1 := \frac{\alpha_0}{q} = \frac{1}{2}(\alpha_0 + 3\alpha) \in (3\alpha,\alpha_0)$. Hence, 
\begin{align} 
     \E\Big(  |G^{B_R}(x,y;z)|^s\, \chi_{\Omega_B}\Big) \le  C\,   R^{- \alpha_1 } . \label{eqn:bad-est}
\end{align}

Suppose $E \le \frac{1}{2}R^{-\delta}$. Then, on the ``good set'' $\Omega_G$, ${\rm dist}(E,\sigma(H^{B_R})) \ge \frac{1}{2}R^{-\delta}$, i.e., $E$ is non-resonant with the eigenvalues of $H^{B_R}$. We can then bound the Green’s function using the Combes--Thomas estimate \eqref{eqn:CT-est} as follows: for $z=E+i\varepsilon$ and $c_M=\frac{\log 2}{2M}$ with $M$ the maximum vertex degree as in \eqref{eqn:bded-geo},
\begin{align*}
   |G^{B_R}(x,y;z)|\le &  \frac{2}{\dist(E,\sigma(H^{B_R}))}e^{-c_M\,\dist(E,\sigma(H^{B_R})) \, d(x,y)} \\
   \le & 4 R^{ \delta} \exp\left(-c_M\,\frac{1}{2}R^{-\delta}\, d(x,y)\right). 
\end{align*}
Hence, for $d(x,y) \ge R/2$, the second term in \eqref{eqn:good-bad-Gs} can be bounded as
\begin{align}
    \E\Big(  |G^{B_R}(x,y;z)|^s\, \chi_{\Omega_G}\Big) \le & 4^s R^{    s\delta } \exp \left(-s\,\frac{c_M}{2}R^{-\delta}\, \frac{1}{2}R\right) \notag \\
    =& 4^s R^{    s\delta } \exp \left(- \,\frac{sc_M}{4}R^{ 1-\delta}\,  \right). \label{eqn:good-est}
\end{align}
Combining \eqref{eqn:bad-est} and \eqref{eqn:good-est} gives
\begin{align}
     \E\Big(  |G^{B_R}(x,y;z)|^s\Big) \le 2C  \, R^{-\alpha_1},
\end{align}
for $R > R_0(s,\delta,\alpha_1,M)$; since $1-\delta>0$, this completes the proof of \eqref{eqn:initial-Gs}.
\end{proof}

\subsection{Geometric Decoupling}
In this section, we show that the Lifshitz tail estimate \eqref{eqn:initial-Gs} on a finite-volume ball leads to a fractal moment bound for the Green's function on the entire space.

\begin{lemma}[Initial-Induction]\label{lem:ini-induc}
Let $\alpha_1 \in (3\alpha,\alpha_0)$, $s \in (0,\tau)$ (sufficiently small), $R > 0$ (sufficiently large), and $E < \frac{1}{2}R^{-\delta}$ be as given in Lemma~\ref{lem:initial}, so that \eqref{eqn:initial-Gs} holds for some $C > 0$.
Fix $u_0 \in \V$, $R > 0$, and let $B_R = B(u_0,R)$. 
Then there exists a constant $C_2 > 0$
such that for any $\eps > 0$ and any $y \in \V$ with $d(u_0,y) \ge R+2$, one has
\begin{align}\label{eqn:induc-step0}
\E \big|G(u_0,y;z)\big|^s
\le C_2 R^{\alpha-\alpha_1} \cdot \sum_{\substack{v' \in \V \\ d(v',u_0) = R+2}} \E \big|G^{B_{R+1}^c}(v',y)\big|^s.
\end{align}
\end{lemma}

To prove Lemma~\ref{lem:ini-induc}, we use the following geometric decoupling. As before, let $H^{X}$ denote the restriction of $H_\omega$ to a subset $X \subset \V$.
Define $\partial^+ X = \{x \in \V : d(x,X) = 1\}$ as the outer (exterior) boundary of $X$,
$\partial X = \{x \in \V : d(x,X^c) = 1\}$ as the inner (interior) boundary,
and 
\begin{align*}
    \Gamma (X)  = \big\{\ (u,v) \text{ or }(v,u) \in \partial X \times \partial^+ X \ \text{with}\ d(u,v) = 1\ \big\}
\end{align*} as the boundary edges.  
For $R > 0$, define
\begin{align}\label{eqn:HR-dsum}
H^{(R)} = H^{B_R} \oplus H^{B_R^c}.
\end{align}
We express $H_\omega$ as the sum of $H^{(R)}$ and the remaining interaction between $B_R$ and $B_R^c$. 
\begin{align*}
    H=H^{(R)}+T^R
\end{align*}
where 
\begin{align*}
   \bra{u} T^R\ket{v}=\begin{cases}
        -1, \ \ &(u,v)\in \Gamma(B_R) \\
        0, \ \ &{\rm  otherwise}
    \end{cases}.
\end{align*}

Denote by $G^{(R)}$ the resolvent of $H^{(R)}$. From \eqref{eqn:HR-dsum}, we have
\begin{align}\label{eqn:GR-dsum}
G^{(R)} = G^{B_R} \oplus G^{B_R^c}.
\end{align}
The second resolvent identity (see, e.g., \cite[Equation (11.10)]{aizenman2015random}) gives
\begin{align}\label{eqn:GR-decomp}
G(x,y;z) = G^{(R)}(x,y;z) - G^{(R)} T^R G(x,y;z)
\end{align}
and
\begin{align*}
G(x,y;z) = G^{(R+1)}(x,y;z) - G T^{R+1} G^{(R+1)}(x,y;z).
\end{align*}
Combing the two identities gives
\begin{align}\label{eqn:resolv-2}
   G(x,y;z)=G^{(R)}(x,y;z)-G^{(R)}T^RG^{(R+1)}(x,y;z)+G^{(R)}T^R GT^{R+1}G^{(R+1)}(x,y;z). 
\end{align} 
For $x = u_0$ and $d(u_0,y) \ge R+2$, we have $x \in B_R$ and $y \in B_{R+1}^c$. Then the first term vanishes due to \eqref{eqn:GR-dsum}:
\begin{align*}
    G^{(R)}(u_0,y;z)=0.
\end{align*}
The second term also vanishes since
\begin{align}
  G^{(R)}T^RG^{(R+1)}(u_0,y;z)=&\sum_{u,v}
  \bra{u_0}  G^{(R)} \ket{u} 
  \bra{u}  T^{R} \ket{v} 
  \bra{v}  G^{(R+1)} \ket{y} \notag \\
   =&-\sum_{(u,v)\in \Gamma(B_R)}
  \bra{u_0}  G^{(R)} \ket{u}  
  \bra{v}  G^{(R+1)} \ket{y},\label{eqn:temp330}
\end{align} 
 where we have dropped the implicit parameter $z$ for notational convenience.
Recall the definition of a boundary pair: $(u,v) \in \Gamma(B_R)$ means either $u \in B_R$ and $v \in B_{R+1}$, or $u \in B_{R+1}$ and $v \in B_R$. In either case, $v \in B_{R+1}$. Since $y \in B_{R+1}^c$, \eqref{eqn:GR-dsum} (with $R+1$) implies that \eqref{eqn:temp330} vanishes.

The remaining term in $G$ is 
\begin{align}
    G(u_0,y;z)=& 
    \bra{u_0}{G^{(R)}T^R GT^{R+1}G^{(R+1)}}\ket{y}
\notag    \\
    =&\sum_{(u,u')\in \Gamma (B_R)}\  \sum_{(v,v')\in \Gamma( B_{R+1})} \bra{u_0} G^{(R)} \ket{u}\bra{u'} G \ket{v}\bra{v'} G^{(R+1)} \ket{y} \notag \\
     =&\sum_{(u,u')\in \Gamma( B_R)} \ \sum_{(v,v')\in  \Gamma( B_{R+1})} \bra{u_0} G^{B_R} \ket{u}\bra{u'} G \ket{v}\bra{v'} G^{B_{R+1}^c} \ket{y} .\label{eqn:3.15}
\end{align}
In the last line, we used $u_0,u \in B_R \Longrightarrow  \bra{u_0} G^{(R)} \ket{u}=\bra{u_0} G^{B_R} \ket{u}$ and $v',y\in B_{R+1}^c   \Longrightarrow  \bra{v'} G^{(R+1)} \ket{y}=\bra{v'} G^{B_{R+1}^c} \ket{y}$.
Then, for $s<1$,
\begin{align}\label{eqn:temp333}
    |G(u_0,y;z)|^s\le \sum_{(u,u')\in \nabla B_R} \sum_{(v,v')\in \nabla B_{R+1}} \Big|\bra{u_0} G^{B_R} \ket{u}\bra{u'} G \ket{v}\bra{v'} G^{B_{R+1}^c} \ket{y}\Big|^s. 
\end{align}

Taking the expectation and rewriting the bra-ket notation as $\bra{x}A\ket{y} = A(x,y)$, we obtain 
\begin{align}\label{eqn:temp334}
   \E |G(u_0,y;z)|^s\le \sum_{\substack{(u,u')\in \Gamma( B_R)\\ (v,v')\in \Gamma( B_{R+1})}} 
    \E_{\ne u',v}\E_{u',v}\Big( \big| G^{B_R}(u_0,u;z)  G (u',v;z)  G^{B_{R+1}^c}(v',y;z)\big|^s \Big),
\end{align}
where we have written $\E_{\ne u'v}$ to denote the expectation over all random variables $V_\omega(x)$ for $x\ne u',v$, and $\E_{u',v}$ for that over $V_\omega(u')$ and $V_\omega(v)$.
Notice that $G^{B_R}$ and $G^{B_{R+1}^c}$ are Green’s functions restricted to subsets and therefore depend only on the random variables $\{V_\omega(x)\}_{x \in B_R}$
and $\{V_\omega(x)\}_{x \in B_{R+1}^c}$, respectively. Since $u',v \notin B_R$, $G^{B_R}(u_0,u)$ is independent of \(V_\omega(u'),V_\omega(v)\).

Similarly, since $u',v \notin B_{R+1}^c$, $G^{B_{R+1}^c}(v',y)$ is also independent of \(V_\omega(u'),V_\omega(v)\).
Hence, 
\begin{align*}
    \E_{u',v}\Big( \big| G^{B_R}(u_0,u)  G (u',v)  G^{B_{R+1}^c}(v',y)\big|^s \Big)=& \big| G^{B_R}(u_0,u)\big|^s \E_{u',v}\Big( |  G (u',v)\big|^s \Big)|   G^{B_{R+1}^c}(v',y)\big|^s \\
    \le &  C_1  \big| G^{B_R}(u_0,u)\big|^s  \big|   G^{B_{R+1}^c}(v',y)\big|^s     
\end{align*}
where we used the a priori bound $\E_{u',v}|G|^s \le C_1 = C_1(s,\tau)$ from Lemma~\ref{lem:priori-G}. Then 
\begin{align*}
     \E\Big( \big| G^{B_R}(u_0,u)  G (u',v)  G^{B_{R+1}^c}(v',y)\big|^s \Big) 
    =&\E_{\neq u',v}\E_{u',v}\Big( \big| G^{B_R}(u_0,u)  G (u',v)  G^{B_{R+1}^c}(v',y)\big|^s \Big) \\
    \le &C_1  \E_{\neq u',v}\Big( \big| G^{B_R}(u_0,u)\big|^s   \big|   G^{B_{R+1}^c}(v',y)\big|^s \Big) \\
     =& C_1   \E \Big( \big| G^{B_R}(u_0,u)\big|^s  \Big) \E \Big( \big|   G^{B_{R+1}^c}(v',y)\big|^s \Big)  .
\end{align*}
In the last two lines, we used that $G^{B_R}(u_0,u)$ and $G^{B_{R+1}^c}(v',y)$ are independent of each other and of $V_\omega(u'), V_\omega(v)$, which correspond to vertices in $B_{R+1}\setminus B_R$.

Substituting this into \eqref{eqn:3.15},
\begin{align*}
    \E \big|G(u_0,y;z)\big|^s\le \sum_{\substack{(u,u')\in \Gamma( B_R)\\ (v,v')\in \Gamma( B_{R+1})}}C_1   \E \Big( \big| G^{B_R}(u_0,u)\big|^s  \Big) \E \Big( \big|   G^{B_{R+1}^c}(v',y)\big|^s \Big)  .
\end{align*}
Notice that 
$u \in B_R(u_0,R)$ and $d(u_0,u) = R$. Therefore, the initial estimate \eqref{eqn:initial-Gs} applies to $G^{B_R}(u_0,u)$, which gives
\begin{align*}
  & \E \big|G(u_0,y;z)\big|^s \\ 
  \le & C_1\Big(\sum_{ (u,u')\in \Gamma( B_R)   }     CR^{-\alpha_1}  \Big)\sum_{   (v,v')\in\nabla B_{R+1} } \E \Big( \big|   G^{B_{R+1}^c}(v',y)\big|^s \Big)      \\
  \le &  C_1 \Big(   \sup_{\V}\deg(x)\cdot \vol(B_{R}\setminus B_{R-1}  ) \,     CR^{-\alpha_1} \Big)\    \sup_{\V}\deg(x)     \sum_{   \substack{v'\in \V \\   d(v',u_0)=R+2}   } \E \Big( \big|   G^{B_{R+1}^c}(v',y)\big|^s \Big) \\
   \le &  C_2  R^{\alpha-\alpha_1}  \,    \sum_{   \substack{v'\in \V \\   d(v',u_0)=R+2} } \E \Big( \big|   G^{B_{R+1}^c}(v',y)\big|^s \Big),
\end{align*}
where we   used $\vol(B_R\setminus B_{R-1})\le\vol(B_R)\le CR^\alpha$.


\subsection{Depleted Green Function Estimate and the Inductive Step}
Lemma~\ref{lem:ini-induc} implies that $\E|G(x,y;z)|^s \le  \rho \E|G(u_1,y;z)|^s$ for some $\rho<1$ and some $u_1$ located at distance $R+2$ from $x = u_0$, provided $3\alpha<\alpha_1$. To prove Theorem~1, we proceed along a path from $x$ to $y$ in approximately $d(x,y)/R$ steps. An appropriate induction then yields $\E|G(x,y;z)|^s \lesssim \rho^{d(x,y)/R} \approx e^{-c d(x,y)}$, which is the desired decay of the fractional moments. The key ingredient in this induction is the following procedure for the depleted Green’s function.

\begin{lemma}\label{lem:dep-green}
Let $H_\omega$ be a random operator of the form \eqref{eqn:AM} with an i.i.d.\ random potential whose single-site distribution $P_0$ is compactly supported and satisfies \eqref{eqn:tau-holder} for some $\tau \in (0,1]$. Then, for any $s \in (0,\tau/2)$, there exists a constant $C_1$ 
such that for any ball $W = B(x_0,R) \subset \V$, any $z = E + i\eps \in \C \setminus \R$, and any $u,y \in W^c = \V \setminus W$,
    \begin{align}\label{eqn:depleted-G}
        \E\big|G^{W^c}(u,y;z)\big|^s\le \E\big|G(u,y;z)\big|^s+C_1  \sum_{ \substack{v\in \V\\ d(x_0,v)=R }  }\E\big|G(v,y;z)\big|^s
    \end{align}
\end{lemma}

\begin{remark}
Results of this type were first established in \cite[Lemma 2.3]{aizenman01finitevol} for the $\Z^d$ case; see also \cite[Lemma 11.4]{aizenman2015random} for an alternative formulation. The proof for $\Z^d$ uses the a priori bound \eqref{eqn:priori-xy} and geometric decoupling \eqref{eqn:GS-decoup}, and it extends to the general graphs considered in this paper. For completeness, we sketch the proof below, following the outline in \cite[Lemma 11.4]{aizenman2015random}.
\end{remark}

\begin{proof}
 Let $H^{(W)} = H^{W} \oplus H^{W^c}$ and $T^W = H_\omega - H^{(W)}$ be as in \eqref{eqn:HR-dsum}. Then, as in \eqref{eqn:GR-decomp}, we have
    \begin{align*}
            G=G^{(W)}-G^{(W)}T^WG  
         \Longrightarrow    G^{(W)}=G +G^{(W)}T^WG.
    \end{align*}
 For $u,y \in W^c$, observe that $G^{(W)}(u,y;z) = \bra{u} G^{W^c} \ket{y}$, and
    \begin{align*}   \bra{u}G^{W^c}\ket{y}=&\bra{u}G \ket{y}+ \sum_{v,v'}
  \bra{u}  G^{(W)} \ket{v} 
  \bra{v}  T^{W} \ket{v'} 
  \bra{v'}  G  \ket{y}    \\
  =&\bra{u}G \ket{y}+ \sum_{(v,v')\in \Gamma( W) }
  \bra{u}  G^{(W)} \ket{v} 
  \bra{v}  T^{W} \ket{v'} 
  \bra{v'}  G  \ket{y}   \\
  =&\bra{u}G \ket{y}- \sum_{\substack{(v,v')\in \Gamma( W) \\
  v\in W^c, v'\in W
  }}
  \bra{u}  G^{W^c} \ket{v} 
  \bra{v'}  G  \ket{y} . 
    \end{align*}
Then, similar to \eqref{eqn:temp333} and \eqref{eqn:temp334}, we obtain
    \begin{align}
  \E \big|G^{(W)}(u,y;z)\big|^s      \le \E \big|G (u,y;z)\big|^s +  \sum_{\substack{(v,v')\in \Gamma( W) \\
  v\in W^c, v'\in W
  }} \E\Big[
  \big| G^{W^c}(u,v;z) \big|^s  \cdot  
  \big| G(v',y;z) \big|^s \Big] .\label{eqn:temp352}
    \end{align} 
The second term on the right-hand side is estimated using the geometric decoupling inequality \eqref{eqn:GS-decoup}:
\begin{align*}
    \E\Big[
  \big| G^{W^c}(u,v;z) \big|^s  \cdot  
  \big| G(v',y;z) \big|^s \Big]\le C \E\Big[
  \big| G(v',y;z) \big|^s \Big] 
\end{align*}
for some constant $C$ depending only on $s$ and $\tau$. Substituting this into \eqref{eqn:temp352} gives \eqref{eqn:depleted-G} with $C_1 = MC$, where $M$ is the maximal vertex degree as in \eqref{eqn:bded-geo}.

\end{proof}

Applying the above lemma to the right-hand side of \eqref{eqn:induc-step0} inductively, we arrive at:
\begin{proof}[Proof of Theorem~\ref{thm:FMbound}]
 It suffices to prove exponential decay for $d(x,y)$ sufficiently large. 
 We start with $x = u_0$ and $d(x,y) \ge R_0 + 2$ for some $R_0$ to be determined later. Combining the initial induction \eqref{eqn:induc-step0} and the depleted Green’s function estimate \eqref{eqn:depleted-G} with $W = B(u_0,R+1)$, we obtain
\begin{align*}
     & \E \big|G(u_0,y;z)\big|^s \\
   \le &    C_2R^{\alpha-\alpha_1}\cdot   \sum_{  \substack{ v'\in \V \\   d(v',u_0)= R+2 } }\Big(\E\big|G(v',y;z)\big|^s+C_1  \sum_{ d(u_0,v)=R+1}\E\big|G(v,y;z)\big|^s\Big) \\
\le &     C_2 R^{\alpha-\alpha_1}\cdot  c_2 (R+2)^\alpha\Big(\sup_{  \substack{ v'\in \V \\   d(v',u_0)= R+2 } }\E\big|G(v',y;z)\big|^s+C_1 c_2(R+1)^\alpha  \sup_{ d(u_0,v)=R+1}\E\big|G(v,y;z)\big|^s\Big)\\
\le & C_2'\, R^{3\alpha-\alpha_1}\,   \cdot \sup_{ d(u_0,u_1)\le R+2}\E\big|G(u_1,y;z)\big|^s ,\end{align*}
where $C_2'$ depends only on $C_2, c_2,$ and $\alpha$. Since \(\alpha_1>3\alpha\), choose $R = R_0$ sufficiently large so that
\begin{align*}
    \rho: =C_2'\, R^{3\alpha-\alpha_1}<1,
\end{align*}
which implies that, for $E < E_0(R_0,\alpha,\beta)$ and $y \in \V$ with $d(u_0,y) \ge R_0 + 2$, 
\begin{align}\label{eqn:u0-u1}
   \E \big|G(u_0,y;z)\big|^s
   \le \rho \, \cdot   \sup_{u_1\in B(u_0,R_0+2)}\E\big|G(u_1,y;z)\big|^s .
\end{align}
Repeating this step for $u_1 \in B(u_0,R_0+2)$ gives
\begin{align}
    \E\big|G(u_1,y;z)\big|^s
   \le \rho   \sup_{u_2\in B(u_1,R_0+2)}\E\big|G(u_2,y;z)\big|^s . \label{eqn:u1u2}
\end{align}
Combining \eqref{eqn:u0-u1} and \eqref{eqn:u1u2}, we obtain 
\begin{align*}
   \E \big|G(x,y;z)\big|^s
   \le &  \rho   \sup_{u_1\in B(x,R_0+2)}\E\big|G(u_1,y;z)\big|^s  \\ 
 \le &  \rho^2   \sup_{u_1\in B(x,R_0+2)}  \Big(\sup_{u_2\in B(u_1,R_0+2)}\E\big|G(u_2,y;z)\big|^s \Big)  \\
 \le & \rho^2  \sup_{u_2\in B(x,2(R_0+2))}\E\big|G(u_2,y;z)\big|^s.
\end{align*}
Inductively, we can apply this process at least $n = \lfloor \frac{d(x,y)}{R_0+2} \rfloor \ge \frac{d(x,y)}{R_0} - 1$ times before we run into the possibility that $d(u_n,y)<R_0+2$. Thus stopping at $n = \lfloor \frac{d(x,y)}{R_0+2} \rfloor$ before this can occur, we obtain
\begin{align*}
     \E \big|G(x,y;z)\big|^s\le  
   \rho^n  \sup_{u_n\in B\big(x,n(R_0+2)\big)}\E\big|G(u_n,y;z)\big|^s  
   \le   \rho^{\frac{d(x,y)}{ R_0}-1}\, C_3,
\end{align*}
where the last term is bounded by the a priori estimate \eqref{eqn:priori-all} with some constant $C_3$ depending only on $s,\lambda,$ and $\tau$. Hence, for all $d(x,y) \ge R_0 + 2$, $E < E_0$, and $z \in \C \setminus \R$,
\begin{align*}
     \E \big|G(x,y;z)\big|^s \le  \frac{C_3}{\rho}  e^{-\mu d(x,y) }, \ \ \mu=\Big|\frac{\log\rho}{R_0}\Big|,
\end{align*}
which proves the fractional moment bound \eqref{eqn:FM-bound}.
\end{proof}


\section{Lifshitz Tail Estimate}\label{sec:lif}
\subsection{(Modified) Neumann Bracketing and the Upper Bound}
In this part, we show that the Neumann eigenvalue lower bound \eqref{eqn:N-ev-lower} implies the Lifshitz-tail upper estimate. The proofs of \eqref{eqn:exp-lif-upper} and \eqref{eqn:infi-lif-up} rely on the same procedure, namely (modified) Neumann bracketing. We begin with the following covering lemma, a standard result for graphs satisfying the volume-growth condition \eqref{eqn:vol-control}; see, for example, \cite[Lemma 6.2]{barlow2017random}.
\begin{lemma}\label{lem:covering}
Suppose the graph $(\V,\mathcal E)$ satisfies \eqref{eqn:vol-control} for some $\alpha \ge 1$. There exist constants $C_i \in (0,\infty)$, depending only on $c_1,c_2$ in \eqref{eqn:vol-control}, such that for any $B(x,R)$ and any $r\le R$, there is a finite cover $\bigcup_{i\in I} B(x_i,r) \supset B(x,R)$ with the following properties:
    \begin{enumerate}
        \item The balls $B(x_i,r/2)$ are disjoint
        \item Every point $y\in B(x,R)$ is covered by at most $C_1$   balls $B(x_i,r)$
        \item    The number of balls in the cover, denoted by $|I|$, satisfies 
        \begin{align}
  C_2 \frac{R^\alpha}{r^\alpha}\le          |I|\le C_3\frac{R^\alpha}{r^\alpha}. 
        \end{align} 
 
    \end{enumerate}
\end{lemma}
 The lower bound in Part (3) follows by a volume comparison. For $r\le R$, the upper bound in Part (3) follows from the inclusion $\bigsqcup_{i\in I}B(x_i,r/2)\subseteq B(x,2R)$. Note that $C_3$ may depend on $\alpha$.

We now employ modified Neumann bracketing for the free Laplacian. Recall that the Neumann Laplacian, viewed as the subgraph Laplacian, satisfies a closed quadratic-form relation (see, e.g., the Discrete Gauss–Green Theorem in \cite[Theorem 1.24]{barlow2017random}): For any $f \in \C^{X}$, let $X$ be either a single ball $B(x_i,r)$ or a cover $U_R:=\bigcup_{i \in I} B(x_i,r)$ of a ball $B(x,R)$ as in Lemma~\ref{lem:covering}. Then
\begin{align}\label{eqn:Diri-form}
 \ipc{f}{-\Delta^{X,N}f}_{\ell^2(X)} = \frac{1}{2} \sum_{ \substack{x,y\in X \\ y\sim x}} \big(f(y)- f(x)\big)^2.
\end{align}
Consider the covering $U_R=\bigcup_{i\in I}B(x_i,r)\supseteq B(x,R)$ and let $C_1$ be as in Lemma~\ref{lem:covering}. Then
\begin{align*} 
C_1\ipc{f}{-\Delta^{U_R,N} f} \ge  \sum_{i\in I  }\ipc{f}{-\Delta^{B(x_i,r),N} f}.
\end{align*}
Accounting for the possible overlap of the onsite potential $V_\omega\ge0$ gives
\begin{align}\label{eqn:Neumann-bra}
    C_1\ipc{f}{(-\Delta^{U_R,N}+V_\omega) f} \ge  \sum_{i\in I  }\ipc{f}{(-\Delta^{B(x_i,r),N}+V_\omega \one_{B(x_i,r)}}) f .
\end{align}

Denote by $E_0(X)$ and $E_1(X)$ the smallest and second smallest eigenvalues of an operator $X$, respectively. Let $H^{X}$ and $H^{X,N}$ denote the restriction of $H$ to $X$ with Dirichlet and Neumann boundary conditions, respectively.
Suppose $f_0 \in \ell^2(U_R)$ is the ground state of $H^{U_R,N}$. Then \eqref{eqn:Neumann-bra} gives
    \begin{align} \label{eqn:I-bound}
    C_1 E_0\big(   H^{U_R,N}\big)=C_1\ipc{f_0}{H^{U_R,N} f_0} \ge \sum_{i \in I} \ipc{f_0}{H^{B(x_i,r),N} f_0}  .
    \end{align}

Notice that the inner product on the right-hand side is taken over each $\ell^2(B(x_i,r))$, so that we can replace $f_0$ with its projection $P_if_0$ onto $B(x_i,r)$.
Since $\bigcup_{i \in I} B(x_i,r) = U_R$, there is at least one $i \in I$ such that $P_i f_0$ is nonzero. This implies there is at least one $i$ such that
\begin{align*}
    \ipc{f_0}{H^{B(x_i,r),N} f_0}=\ipc{P_if_0}{H^{B(x_i,r),N} P_if_0}\ge E_0(H^{B(x_i,r),N}). 
\end{align*}
Hence, by the min–max principle (see, e.g., \cite[Theorem 4.13]{aizenman2015random}),
 \begin{align*} 
    C_1 E_0\big(   H^{U_R,N}\big) \ge \inf_{i \in I} E_0\big(H^{B(x_i,r),N}\big).
    \end{align*}
Also, since $\bigcup_{i \in I} B(x_i,r) = U_R$ covers $B(x,R)$, applying the min–max principle (or the Cauchy interlacing theorem) to $H^{U_R,N}$ and its (Dirichlet) restriction to $B(x,R) \subset U_R$ gives
\begin{align}
E_0\big(   H^{B(x,R)}\big) \ge     E_0\big(   H^{U_R,N}\big) .
\end{align}
Then for any $\gamma>0$, 
\begin{align}
    \P\Big(E_0\big(  H^{B(x,R)}\big)\le \gamma \Big) \le & \P\Big(\inf_{i \in I} E_0(H^{B(x_i,r),N})\le \frac{1}{C_1}\gamma \Big) \notag \\
    \le & |I|\max_i\P\Big( E_0(H^{B(x_i,r),N})\le \frac{1}{C_1}\gamma \Big) \notag \\
    \le & C_3\frac{R^\alpha}{r^\alpha}\max_i\P\Big( E_0(H^{B(x_i,r),N})\le \frac{1}{C_1}\gamma \Big).  \label{eqn:A10}
\end{align}

To prove \eqref{eqn:exp-lif-upper}, it suffices to obtain a sub-exponential bound for $\P\Big( E_0(H^{B(x_i,r),N})\le R^{-\delta} \Big)$ for an appropriate choice of $R$ and $r$. The key ingredients are Temple’s inequality and a large-deviation estimate, both classical.
\begin{proposition}[Temple \cite{temple1928theory}] \label{prop:temple}
    Let $H$ be a self-adjoint operator with an isolated non-degenerate eigenvalue $E_0=\inf \sigma(H)$, and let $E_1=\inf \big(\sigma(H) \backslash\{E_0\}\big)$
.  Then for any
$\psi\in \mathcal D (H)$  (domain of $H$), which satisfies $\ipc{\psi}{H\psi}<E_1$, $\|\psi\|=1$, the following bound holds:
\begin{align}\label{eqn:temple}
    E_0\ge \ipc{\psi}{H\psi}-\frac{\ipc{H\psi}{H\psi}-\ipc{\psi}{H\psi}^2}{E_1-\ipc{\psi}{H\psi}}. 
\end{align}
\end{proposition}

\begin{proposition}[Hoeffding {\cite{hoe1963}}]\label{prop:hoeffding}
           If $\{Y_k\}_{1\le k\le K}$ are i.i.d. random variables ranging in $[0,b]$, then for any $\eps>0$, 
\begin{align}
            \P\Big(\frac{1}{K}\sum_{k=1}^KY_k\le \E(Y_k)-\eps\Big)\le e^{-2K\varepsilon^2/b^2}.  \label{eqn:ldt}
        \end{align} 
    \end{proposition}

Given a ball $B_r = B(x_i, r) \subset \V$ for some $x_i \in \V$ and $r > 0$, let $\Delta^{B_r,N}$ denote the Neumann Laplacian on $B_r$, defined in \eqref{eqn:N-lap}. From \eqref{eqn:Diri-form}, the smallest eigenvalue of $-\Delta^{B_r,N}$ satisfies  $E_0(-\Delta^{B_r,N})=0$. 
 Let $0 < E_1(-\Delta^{B_r,N})$ be the second smallest eigenvalue (i.e., the first non-zero eigenvalue) of $-\Delta^{B_r,N}$.
Following Assumption~\ref{assume:N-ev}, suppose that for some $\beta \ge 2$, $c_0 > 0$, and $r > 0$, and for any ball $B_r=B(x_i,r)$ in Lemma~\ref{lem:covering},
\begin{align}\label{eqn:A4}
    E_1(- \Delta^{B_r,N}) \ge \frac{c_0}{r^\beta}. 
\end{align} 
We define a truncated potential by
    \begin{align}\label{eqn:wtV-def}
        \wt V(x):=\min \left\{ V_\omega(x), \frac{c_0}{3}r^{-  \beta }\right\}.
    \end{align} 
Let $\wt H^{r} := -\Delta^{B_r,N} + \wt V$. Then $\wt H^{r} \le H^{B_r,N}$ by the definition of $\wt V$. By the variational or min-max principle,
    \begin{align}\label{eqn:SGE00}
      E_0\big( \wt H^{ r}\big)\le   E_0\big(   H^{B_r,N}\big).
    \end{align}
 Combining the min–max principle with the bound $\wt H^{r} \ge -\Delta^{B_r,N}$ and the lower bound \eqref{eqn:A4}, we also obtain
\begin{align}E_1\big( \wt H^{ r}\big)\ge E_1(-\Delta^{B_r ,N})\ge  \frac{c_0}{r^{  \beta }}. 
\end{align}

Next we bound $E_0\big( \wt H^{r} \big)$ from below using Temple’s inequality. Let $\psi(x) = |B_r|^{-1/2}$ for $x \in B_r$ be the normalized constant ground state of the Neumann Laplacian $-\Delta^{B_r,N}$; it satisfies $\Delta^{B_r,N}\psi = 0$ and
\begin{align*}
       \ipc{\psi}{\wt H^{ r}\psi}=\ipc{\psi}{\wt V \psi}= \frac{1}{|B_r |}\sum_{x\in B_r }\wt V(x)\le \frac{c_0}{3}r^{-  \beta }< E_1\big( \wt H^{ r}\big).
     \end{align*}
Thus the conditions of Temple’s inequality are satisfied with $H=\wt H^r$, $E_0 = E_0\big( \wt H^{r} \big)$, $E_1 = E_1\big( \wt H^{r} \big)$, and $\psi$ the normalized ground state of the Neumann Laplacian.

The second term on the right-hand side of Temple’s inequality \eqref{eqn:temple} can be estimated as
     \begin{align}
   \frac{\ipc{\wt H^{ r}\psi}{\wt H^{ r}\psi}-\ipc{\psi}{\wt H^{ r}\psi}^2}{E_1-\ipc{\psi}{\wt H^{ r}\psi}}  \le  &\,     \frac{\ipc{\wt H^{ r}\psi}{\wt H^{ r}\psi} }{ E_1-\ipc{\psi}{\wt H^{ r}\psi}} 
   =  \frac{\ipc{\wt V \psi}{\wt V \psi} }{ E_1-\ipc{\psi}{\wt H^{ r}\psi}}\nonumber \\
   \le &\, \frac{\frac{c_0}{3}r^{-  \beta }|B_r |^{-1}\sum_{x\in B_r }\wt V(x) }{ c_0r^{-  \beta }-\frac{c_0}{3}r^{-  \beta }}=\frac{1}{2|B_r |}\sum_{x\in B_r }\wt V(x).  \label{eqn:5.14}
     \end{align}
Substituting \eqref{eqn:5.14} into \eqref{eqn:temple}, we obtain
     \begin{align}
      E_0\big( \wt H^{ r}\big)\ge &\,    \ipc{\psi}{\wt H^{ r}\psi}-\frac{\ipc{\wt H^{ r}\psi}{\wt H^{ r}\psi}-\ipc{\psi}{\wt H^{ r}\psi}^2}{E_1-\ipc{\psi}{\wt H^{ r}\psi}} \notag \\
  \ge &\,     \frac{1}{|B_r |}\sum_{x\in B_r }\wt V(x)-\frac{1}{2|B_r |}\sum_{x\in B_r }\wt V(x)=\frac{1}{2|B_r |}\sum_{x\in B_r }\wt V(x). \label{eqn:Temple-app}
     \end{align}
Observe that $\{ r^\beta  \wt V(x) \}_{x \in B_r}$ are i.i.d.\ random variables with range in $[0, c_0/3]$ and common mean
\begin{align*}
    \mu_r=\E\Big(\min\big\{r^{  \beta}V_\omega(x),\ \frac{c_0}{3} \big\}\Big)\ge \frac{c_0}{3}\P\big(r^{  \beta}V_\omega(x)>\frac{c_0}{3}\big)=\frac{c_0}{3}\Big[1-\P\big(V_\omega(x)\le\frac{c_0}{3 r^{  \beta}}\big)\Big]. 
\end{align*}
Therefore, 
\begin{align}\label{eqn:inf-mu-0}
    \varliminf_{r\to \infty} \mu_r\ge \frac{c_0}{3}[1-\P(V_\omega(x)=0)]=:\frac{c_0}{3}p_1>0, 
\end{align}
since $p_0 = 1 - p_1 = \P(V(x) = 0) < 1$ (since the distribution is non-trivial, i.e., its support contains more than one point). Then for $r$ sufficiently large (depending on $\beta$ and $p_0$),
\begin{align}\label{eqn:inf-mu}
     \mu_r\ge  \frac{c_0}{4}p_1>0, 
\end{align}

\begin{proof}[Proof of Theorem~\ref{thm:lif}]
    
For any $\delta \in (0,1)$, and for $C_1$ specified in \eqref{eqn:A10}, $c_0,r$ specified in \eqref{eqn:A4}, and  $p_1$ specified in \eqref{eqn:inf-mu}, define
\begin{align}\label{eqn:A23}
    R=  \left(\frac{ 12}{C_1c_0p_1}\right)^{\frac{1}{\delta}} r^{\frac{\beta
    }{\delta}}, 
\end{align}
 so that 
 \begin{align}\label{eqn:A24}
     \frac{2}{C_1}r^\beta R^{-\delta}= \frac{c_0p_1}{6}, \qquad {\rm and} \ \ r= \left(\frac{C_1c_0p_1}{ 12}\right)^{\frac{1}{\beta}}R^{\frac{\delta}{\beta}}.
 \end{align}
 Notice that the definition of graph balls in \eqref{eqn:vol-control} allows the radius $R$ to take non-integer values.

Combining \eqref{eqn:A10}, \eqref{eqn:SGE00}, and \eqref{eqn:Temple-app}, we obtain
\begin{align}
  \P\Big(E_0\big(  H^{B(x,R)}\big)\le R^{-\delta}\Big) \le &   C_3 R^{\alpha} \P\Big( E_0(H^{B(x_i,r),N})\le \frac{1}{C_1}R^{-\delta} \Big)\label{eqn:5.21}\\
  \le & C_3 R^{\alpha} \P\Big( E_0(\wt H^{r})\le \frac{1}{C_1}R^{-\delta} \Big) \notag \\ 
 \le & C_3 R^{\alpha}
   \P \Big( \frac{1}{2|B_r |}\sum_{x\in B_r }\wt V(x)\le \frac{1}{C_1}R^{-\delta} \Big) \notag \\  
\le &  C_3 R^{\alpha} \P \Big( \frac{1}{|B_r |}\sum_{x\in B_r }r^{  \beta}\wt V(x)\le \frac{2}{C_1}r^{\beta}R^{-\delta} \Big) \notag \\
= &  C_3 R^{\alpha} \P \Big( \frac{1}{|B_r |}\sum_{x\in B_r }r^{  \beta}\wt V(x)\le \frac{c_0p_1}{6} \Big). \label{eqn:A29}
\end{align}
Finally, applying Hoeffding’s inequality \eqref{eqn:ldt} with $\{ Y_k \}_{k \in K} = \{ r^\beta \wt V(x) \}_{x \in B_r}$, $K = |B_r|$, and $\varepsilon = \frac{c_0 p_1}{12}$
to the right-hand side of \eqref{eqn:A29} along with \eqref{eqn:inf-mu} gives
\begin{align*}
  \P \Big( \frac{1}{|B_r |}\sum_{x\in B_r }r^{  \beta}\wt V(x)\le \frac{c_0p_1}{6} \Big) 
  \le &\,\P \Big( \frac{1}{|B_r |}\sum_{x\in B_r }r^{  \beta}\wt V(x)-\mu_r\le \frac{c_0p_1}{6}-\frac{c_0p_1}{4} \Big) \\
  \le &\,\P \Big( \frac{1}{|B_r |}\sum_{x\in B_r }r^{  \beta}\wt V(x)-\mu_r\le  -\frac{ c_0p_1}{12} \Big) \\
  \le &\, e^{-c|B_r |},
\end{align*}
where $c > 0$ depends only on $c_0$ and $p_1$.

Using the volume lower bound $|B_r| \ge c_1 r^\alpha=c_1'R^{\frac{\delta \alpha}{\beta}}$ from \eqref{eqn:A23}, we conclude
\begin{align} \label{eqn:434}
  \P\Big[E_0\big(  H^{B(x,R)}\big)\le R^{-\delta}\big]\le    C_3 R^{\alpha}  \exp\big({-c|B_r |}\big) \le C_3 \exp\big({-\eta R^{\frac{\delta \alpha}{\beta}} }\big)
\end{align}
 for $\eta=c\left(\frac{C_1c_0p_1}{ 12}\right)^{\frac{\alpha}{\beta}}$, which  completes the proof of \eqref{eqn:exp-lif-upper}.

\end{proof}

\begin{proof}[Proof of Theorem~\ref{thm:infi-lif-upper}]
The proof follows the same approach as for \eqref{eqn:exp-lif-upper}. Under the assumptions of Theorem~\ref{thm:infi-lif-upper}, there exists a sequence $r_i$ satisfying $r_{i+1}<c_3r_i$ and Assumption~\ref{assume:N-ev}. Then, for small $E>0$, since $r_i\nearrow\infty$, there exist $r=r_i<r'=r_{i+1}$ in the sequence such that
\begin{align} \label{eqn:r-choice}
    r^\beta<\frac{c_0p_0C_1}{12E}<r'^\beta. 
\end{align}
  Consequently, 
    \begin{align*}
        r^\beta E\le \frac{c_0p_0C_1}{12 },\ \ {\rm and}\ \ r\ge c_4E^{-\frac{1}{\beta}},\ \ c_4=c_3^{-1}(c_0p_0C_1/12)^{1/\beta}.
    \end{align*}
    
Next, for any $B(x,R)$, use this choice of $r$ in \eqref{eqn:r-choice} to define the cover $U_R=\bigcup_{i\in I}B(x_i,r)\supset B(x,R)$ as in \eqref{eqn:Neumann-bra}. Let $\mathcal{N}(E,\cdot)$ denote the eigenvalue counting function in \eqref{eqn:ev-count}. Then \cite[Lemma A.1, Lemma A.2]{shou2024spectrum} imply
\begin{align*}
   \mathcal{N}(E,H^{B(x,R)})\le  \mathcal{N}(E,H^{U_R,N}) \le \sum_{i\in I}\mathcal{N}\left(\frac{1}{C_1}E,H^{B(x_i,r),N}\right).
\end{align*}
Taking expectations, we obtain  
\begin{align}
   \frac{1}{|B(x,R)|}\E \mathcal{N} (E,H^{B(x,R)} )\le &   \frac{1}{|B(x,R)|}  \sum_{i\in I}\E\mathcal{N}\left(\frac{1}{C_1}E,H^{B(x_i,r),N}\right) \notag  \\
\le &   \frac{|I|}{|B(x,R)|}   \max_{i\in I}\E\mathcal{N}\left(\frac{1}{C_1}E,H^{B(x_i,r),N}\right) \notag \\
\le & c_5 \max_{i\in I}\P\left(E_0(H^{B(x_i,r),N})\le \frac{1}{C_1}E\right),  \label{eqn:N-upp-proof}
\end{align}
where $c_5$ is independent of $R,r$. Note that in the last inequality we used the volume bound \eqref{eqn:vol-control}, the upper bound on $|I|$ in \eqref{eqn:I-bound}, and
\[  \mathcal{N}\left( E/C_1,H^{B(x_i,r),N}\right)\le  |B(x_i,r)| \cdot  \one_{ \{E_0 \le E/C_1 \} } \le c_2 r^\alpha \one_{ \{E_0 \le E/C_1 \} }. \]

Repeating the argument of \eqref{eqn:5.21}–\eqref{eqn:434}, with \eqref{eqn:A10} replaced by \eqref{eqn:N-upp-proof}, we obtain, for any $R>0$ and $E>0$,
\begin{align*} 
\frac{1}{|B(x,R)|}\E \mathcal{N} (E,H^{B(x,R)} )\le &   c_5 \max_{i\in I}\P\left(E_0(H^{B(x_i,r),N})\le \frac{1}{C_1}E\right)\\ 
\le &   c_5 \max_{i\in I} \P \Big( \frac{1}{|B(x_i,r) |}\sum_{x\in B(x_i,r) }r^{  \beta}\wt V(x)\le \frac{2}{C_1}r^{\beta}E \Big) \\
\le& c_5 \max_{i\in I} \P \Big( \frac{1}{|B(x_i,r) |}\sum_{x\in B(x_i,r) }r^{  \beta}\wt V(x)\le \frac{c_0p_1}{6} \Big)\\ 
\le& c_5 \max_{i\in I} e^{-c|B(x_i,r) |} \\
\le& c_5   \exp({-c_6 E^{  -\frac{\alpha}{\beta}}  }),
\end{align*}
for  $c_5,c_6$ depending only on $c_0,p_0,\alpha,\beta$. For fixed $E$, taking $\limsup_{R\to \infty}$ completes the proof of \eqref{eqn:infi-lif-up}. The smallness condition on $E$ ensures that the corresponding $r\propto E^{-1/\beta}$ is sufficiently large so that $\mu_r\ge \frac{c_0p_1}{4}$ in \eqref{eqn:inf-mu}.
\end{proof}

\subsection{(Modified) Dirichlet bracketing and the lower bound}
For the lower bound, one can also obtain a finite-volume probability estimate for the ground-state energy as in \eqref{eqn:exp-lif-upper}. Since such a finite-volume probability lower bound is not required for the Anderson localization result, we omit it and focus on the lower bound for the integrated density of states \eqref{eqn:infi-lif-lower}.

 Let $r_i\nearrow\infty, r_{i+1}\le c_3'r_i$ be given as in Theorem~\ref{thm:infi-lif-lower}. Then, similarly to \eqref{eqn:r-choice}, for any $E>0$, pick consecutive $r'<r$ from the list such that
 \begin{align*} 
    r'^\beta<\frac{2c_0' }{ E}<r^\beta. 
\end{align*}
   Consequently, $r$ satisfies  
    \begin{align}\label{eqn:r-bound-Diri}
      \frac{1}{  r^\beta} \le \frac{E}{2c_0'} ,\ \ {\rm and}\ \ r\le c'_4E^{-\frac{1}{\beta}},\ \ c'_4=c_3' (2c_0 )^{1/\beta}.
    \end{align}

For any $R>4r$, consider the covering of $B(x,R/2)$ given by Lemma~\ref{lem:covering} using balls of radius $4r$, i.e.,
   \begin{align*}
     B(x,R/2)\subset  \bigcup_{i\in \wt I }B(x_i,4r)\subset  B(x,R).  
   \end{align*}
  It follows from Lemma~\ref{lem:covering} that the balls $B(x_i,2r)$ are disjoint and 
     \begin{align}\label{eqn:wtI}
  C'_2 \frac{R^\alpha}{r^\alpha}\le          |\wt I|\le C'_3\frac{R^\alpha}{r^\alpha}. 
        \end{align} 
Let $Q=B(x,R)\setminus \Big(\bigcup_{i\in \wt I}B(x_i,2r)\Big)$. Since $B(x_i,2r)$ are disjoint, 
\begin{align}\label{eqn:parti}
   B(x,R)=Q\cup\bigcup_{i\in \wt I}B(x_i,2r)
\end{align}
is a disjoint partition of $B(x,R)$.  

To bound the quadratic energy of the Dirichlet Laplacian on $B_R(x)=B(x,R)$ from above, we use the modified Dirichlet Laplacian on a subset $A\subset \V$, defined as:
    \begin{align}\label{eqn:Lap-D}
        -\Delta^{A,D}f(x)= -\Delta^{A}f(x)+\deg_{\V\backslash A}(x)\cdot f(x),   \   f\in \ell^2(A), 
    \end{align}
  which is a Dirichlet-type operator with zero/simple boundary conditions on $\V\backslash A$, modified at the interior boundary vertices of $A$ to penalize such boundary vertices.

We consider the modified Dirichlet Laplacian on each $B_{2r}(x_i)=B(x_i,2r)$ and on $Q$. Then
\begin{align*}
    \ipc{f}{-\Delta^{B_R(x)}f}_{\ell^2(B_R(x))}\le \ipc{f}{-\Delta^{Q,D}f}_{\ell^2(Q)}+\sum_{i\in \wt I}\ipc{f}{-\Delta^{B_{2r}(x_i),D}f}_{\ell^2(B_{2r}(x_i))}.
\end{align*}
Adding the potential gives 
\begin{align*}
    \ipc{f}{H^{B_R(x)}f}_{\ell^2(B_R(x))}\le \ipc{f}{H^{Q,D}f}_{\ell^2(Q)}+\sum_{i\in \wt I}\ipc{f}{H^{B_{2r}(x_i),D}f}_{\ell^2(B_{2r}(x_i))},
\end{align*}
where $H^{B_R(x)}=-\Delta^{B_R(x)}+V_\omega$ and $H^{A,D}=-\Delta^{A,D}+V_\omega$ for $A=B(x_i,2r)$ or $Q$. 

It follows from the disjoint partition in \eqref{eqn:parti} and \cite[Lemma A.3]{shou2024spectrum} that 
\begin{align}\label{eqn:456}
    \mathcal{N}(E;H^{B_R(x)})\ge   \mathcal{N}(E;H^{Q,D}) \, +\, \sum_{i\in \wt I} \mathcal{N}(E;H^{B_{2r}(x_i),D})\ge |\wt I|\min_{i\in \wt I} \mathcal{N}(E;H^{B_{2r}(x_i),D}).
\end{align}

Let $E_0=E_0\big(H^{B_{2r}(x_i),D}\big)$ be the ground-state energy of $H^{B_{2r}(x_i),D}$. Given $E>0$, if 
\[E_0\big(H^{B_{2r}(x_i),D}\big)\le E,\]  then $\mathcal N(E; H^{B_{2r}(x_i),D})$ is at least one. Hence,
\begin{align*}  
     \E \mathcal N(E; H^{B_{2r}(x_i),D})\ge \P\Big(E_0(H^{B_{2r}(x_i),D})\le E\Big).
\end{align*}
Combining \eqref{eqn:456} with \eqref{eqn:wtI}, we obtain 
\begin{align}
\frac{1}{|B_R(x)|}\E  \mathcal{N}\Big(E;H^{B_R(x)}\Big) \ge &    \frac{|\wt I|}{|B_R(x)|} \min_{i\in \wt I} \E \mathcal N\Big(E; H^{B_{2r}(x_i),D}\Big)\notag \\
\ge & c_5'r^{-\alpha}\min_{i\in \wt I}\P\Big(E_0(H^{B_{2r}(x_i),D})\le E\Big). \label{eqn:459}
\end{align}
For each $B_r(x_i)\subset B_{2r}(x_i)$, let $\varphi_0$ be the ground state of $-\Delta^{B_r(x_i)}$ on $B_r(x_i)$, with zero extension to $B_{2r}(x_i)$. By the min–max principle,    
\begin{align*}
   E_0(H^{B_{2r}(x_i),D})=\inf_{\varphi\neq 0}\frac{\ipc{\phi}{H^{B_{2r}(x_i),D}\phi}}{\ipc{\phi}{\phi}} \le & \,\inf_{\varphi\neq 0}\frac{\ipc{\phi}{-  \Delta^{B_{2r}(x_i)}\phi}}{\ipc{\phi}{\phi}}+\max_{B_{2r}(x_i)}V_\omega(x) \nonumber\\
    \le & \, \frac{\ipc{\phi_0}{-  \Delta^{B_{2r}(x_i)}\phi_0}}{\ipc{\phi_0}{\phi_0}}+\max_{B_{2r}(x_i)}V_\omega(x) \nonumber\\
       = &\, E_0\big(-\Delta^{B_r(x_i)}\big)+\max_{B_{2r}(x_i)}V_\omega(x)\\
     \le  &\, \frac{c_0'}{r^\beta}+\max_{B_{2r}(x_i)}V_\omega(x)  \\
       \le  &\, \frac{E}{2}+\max_{B_{2r}(x_i)}V_\omega(x)  . 
\end{align*}
In the last two lines, we used the upper bound \eqref{eqn:D-ev-upper} for the Dirichlet eigenvalue and the choice of \(E\) and \(r\) in \eqref{eqn:r-bound-Diri}.

Hence, for sufficiently small $E$,
\begin{align}
  \P\Big(E_0(H^{B_{2r}(x_i),D})\le E\Big)\ge \P\Big(\max_{B_{2r}(x_i)}V_\omega(x)\le \frac{E}{2}\Big) \ge&\;  \P\Big({\textrm{For all} } \ x\in B_{2r}(x_i), V_\omega(x)\le \frac{E}{2}\Big)\nonumber \\
 \ge &\; C\big(\frac{E}{2}\Big)^{\kappa | B_{2r}(x_i)|} \notag \\
  \ge &\; Ce^{-c_6 |\log E| E^{-\frac{\alpha}{\beta}}}. \label{eqn:PE0-lower}
\end{align}
In the last three lines we use the distributional assumption $\P\big(V(x)\le E\big)\ge CE^{\kappa}$ from \eqref{eqn:V-Lif-ass} together with the upper volume bound
\[| B_{2r}(x_i)|\le  c_2r^\alpha=c_2(c_4')^\alpha E^{-\frac{\alpha}{\beta}}  \]
from \eqref{eqn:vol-control} and \eqref{eqn:r-bound-Diri}. We also use the elementary inequality $2\log(E)<\log(E/2)<0$ for small $E$. Combining \eqref{eqn:PE0-lower} with \eqref{eqn:459}, we obtain
\begin{align*}
\frac{1}{|B_R(x)|}\E  \mathcal{N}(E;H^{B_R(x)}) \ge   c_5'r^{-\alpha} Ce^{-c_6 |\log E|   E^{-\frac{\alpha}{\beta}}} 
\ge   c_6' e^{-c_7 |\log E | E^{-\frac{\alpha}{\beta}}}.
\end{align*}
This proves \eqref{eqn:infi-lif-lower}. 

As in the upper bound, for fixed $E$, taking the limit inferior as $R\to \infty$ and then the double logarithmic limit as $E\searrow 0$ produces the log–log lower bound in \eqref{eqn:lif-log-log}
        \[ \lim_{E\searrow 0} \varliminf_{R\to \infty}  \frac{1}{\log E}\log \Big|\log \frac{\E  \mathcal{N}(E;H^{B_R(x)})}{|B_R(x)|}\Big|\ge -\frac{\alpha}{\beta}. \]


\appendix

\section{Combes--Thomas estimates, A priori bounds, and Decoupling inequalities }\label{app:pB-CT}

\begin{lemma}[Combes--Thomas estimates]\label{lem:CT}
Let $(\V,\mathcal E)$ be a graph with a metric $d$ and satisfying the bounded geometry condition \eqref{eqn:bded-geo} with a constant $M>0$. Let $H_\omega=-\Delta+V_\omega$ be as in \eqref{eqn:AM} acting on $\ell^2(\V)$. For any $x\in\V$ and $R>0$, let $H^{B_R}$ be the restriction of $H_\omega$ on a finite ball $B_R=B(x,R)$. For $z\in \C$, let $\mu=\dist\big(z,\sigma(H^{B_R})\big)$. If $0<\mu \le 2M$, then the Green's function of $H^{B_R}$ satisfies 
 \begin{align}\label{eqn:CT-est}
      |G^{B_R}(x,y;z)|\le \frac{2}{\mu} e^{-\frac{\log 2}{2M}\, \mu \, d(x,y)}.
 \end{align}
\end{lemma}

\begin{remark}    
This follows from a more general Combes--Thomas estimate in \cite[Theorem 10.5]{aizenman2015random}: 
\begin{align*}
      |G^{B_R}(x,y;z)|\le \frac{1}{\mu-S_\alpha} e^{-  \alpha \, d(x,y)}. 
 \end{align*}
for any $S_\alpha:=\sup_x\sum_y|H(x,y)|(e^{\alpha d(x,y)}-1)<\mu$. 
For the Schr\"odinger operator $H_\omega$ in \eqref{eqn:AM},
we can evaluate $S_\alpha\le M(e^\alpha-1)$ for any $\alpha>0$. 
The condition $S_\alpha<\mu$ is satisfied by taking
\begin{align*}
\alpha&= \log\left(\frac{\mu}{2M}+1\right),
\end{align*}
which ensures $S_\alpha\le\frac{1}{2}\mu$.
Using that $\log(1+x)\ge (\log 2)x$ for $0\le x\le1$, we obtain that for $0<\mu\le 2M$,
\[\alpha\ge \frac{\log 2}{2M}\, \mu. \]
Thus, 
\begin{align*}
|G^{B_R}(x,y;z)|&\le\frac{2}{\mu}\exp\left[-\frac{\log 2}{2M}\, \mu\,  d(x,y)\right].
\end{align*}
 
\end{remark}

\begin{lemma}[A priori bound on fractional moments]\label{lem:priori-G}
For any graph $(\V,\mathcal E)$, let $H_\omega=-\Delta+V_\omega$ be as in \eqref{eqn:AM} acting on $\ell^2(\V)$.
Assume that the common probability distribution $P_0$ of $V_\omega$ satisfies Assumption~\ref{assume:holder} with some $\tau\in (0,1]$ and $\kappa_\tau\in(0,\infty)$. For any $s\in(0,\tau)$, there exists a constant $C>0$, depending only on $s$, $\|V_\omega\|_\infty$, $\tau$, and $\kappa_\tau$, such that for any subset $\Lambda\subset \V$, any $x,y\in \Lambda$, and any $z\in \C\backslash \R$, the Green's function of $H_\omega$ satisfies
\begin{align}\label{eqn:priori-xy}
     \E_{x,y} |G^\Lambda(x,y;z)|^s \le C,
\end{align}
where $\E_{x,y}(\cdot)=\E(\cdot \, |\, \{V_\omega(u)\}_{u\in \Lambda\backslash \{x,y\}} )$ denotes the conditional expectation with $\{V_\omega(u)\}_{u\in \Lambda\backslash \{x,y\}}$ fixed. As a consequence,    \begin{align}\label{eqn:priori-all}
       \E |G^\Lambda(x,y;z)|^s \le C.
   \end{align}
\end{lemma}
\begin{remark}
Such a priori bounds are closely related to Wegner’s bound \cite{wegner} on the density of states. Fractal bounds of the form \eqref{eqn:priori-all}, especially in the off-diagonal case, were first proved in \cite{AizenmanMolchanov}; see also \cite{graf94,aizenman01finitevol,aizenman2015random} for further developments and discussion. Here, we use the modern version found in \cite[Theorem 8.3, Corollary 8.4]{aizenman2015random}, whose proof employs a modified argument from \cite{graf94}. The dependence of the constant $C$ on the parameters is explicit: $C=\frac{\tau}{\tau-s}\frac{(4\kappa_\tau)^{s/\tau}}{\|V_\omega\|_\infty^s}$, as shown in the cited references.

Finally, the a priori bound for complex energy extends to Lebesgue-almost every $E\in\R$. Moreover, if $\Lambda$ is finite, the extension holds for all real energies $E\in\R$; see \cite[Theorem 8.3, Corollary 8.4]{aizenman2015random} and the discussion thereafter.
\end{remark}

Decoupling expectation values of products of Green’s functions is a key step in the fractional moments method for proving localization. Because the Green’s function is a fractional linear function of individual potential values, suitable regularity conditions on the distribution of $V_\omega$ permit the decoupling of conditional expectations of such products. These techniques are systematically presented in \cite[\S 8.4]{aizenman2015random} for a broad class of random graph operators. Below, we restate relevant results for the Anderson model \eqref{eqn:AM} considered in this work.

\begin{lemma}[{\cite[Corollary 8.11 ]{aizenman2015random}}]
Let $H_\omega$ be a random operator of the form \eqref{eqn:AM} with an i.i.d. random potential whose single-site distribution $P_0$ is compactly supported and satisfies \eqref{eqn:tau-holder} for some $\tau \in (0,1]$. Then, for any $s \in (0,\tau/2)$, there exists $D_s > 0$ such that for any $\Lambda,\Lambda' \subset \V$, and all 
$x,y \in \Lambda$ and $x',y' \in \Lambda'$, 
the conditional averages of the corresponding restricted Green's functions, averaged over a single potential variable $V_\omega(x)$, satisfy:

\begin{align}\label{eqn:decoup-AW8.11}
    \E_{x}\Big[   \big|G^{\Lambda}(x,y;z) \big|^s\big|G^{\Lambda'}(x',y';z) \big|^s     \Big]\le D_s\,    \E_{x}\Big[ \, \big|G^{\Lambda}(x,y;z) \big|^s    \Big] \,    \E_{x}\Big[    \big|G^{\Lambda'}(x',y';z) \big|^s     \Big].
\end{align}
    
\end{lemma}

Applying the decoupling inequality \eqref{eqn:decoup-AW8.11} twice yields the following decoupling lemma.

\begin{lemma}[{\cite[Appendix C]{aizenman01finitevol}}]\label{lem:GS-decoup}
Let $H_\omega$ be a random operator of the form \eqref{eqn:AM} with an i.i.d. random potential whose single-site distribution $P_0$ is compactly supported and satisfies \eqref{eqn:tau-holder} for some $\tau \in (0,1]$. Then, for any $s \in (0,\tau/2)$, there is a constant $C=C(s,\tau,\lambda)$ such that for any $W\subset \V$,  $u,v\in W^c$ and $v',y\in \V$, one has 
 \begin{align}\label{eqn:GS-decoup}
      \E \Big[
  \big| G^{W^c}(u,v;z) \big|^s\, \cdot \, 
  \big| G(v',y;z) \big|^s \Big] \le  C \,  \E \Big[
  \big| G(v',y;z) \big|^s \Big]   .
 \end{align}
\end{lemma}

\begin{proof}[Proof of Lemma~\ref{lem:GS-decoup}]
For $u,v\in W^c$ and $v',y\in W\subset \V$,  let $G_1=G^{W^c}(u,v;z)$ and $G_2=G^{\V}(v',y;z)$. Applying \eqref{eqn:decoup-AW8.11} to $G_1$ and $G_2$ with the expectation over   $V_\omega(v)$ gives 
\begin{align*}
    \E_{v}\big[|G_1|^s\, |G_2|^s\big]\le D_s\E_{v}\big[|G_1|^s \big]\E_{v}\big[  |G_2|^s\big].
\end{align*}
Next, we rewrite $\E_{v}\big[|G_2|^s\big] = \E_{\wt v}\big[|G_2|^s\big]$ as an integral over an independent copy $\wt V_\omega(v)$ of $V_\omega(v)$. Then
\begin{align*}
   \E_{v}\big[|G_1|^s\, |G_2|^s\big]\le    D_s\E_{v}\big[|G_1|^s \big]\E_{\wt v}\big[  |G_2|^s\big] =
      D_s\E_{v,\wt v}\big[|G_1|^s\, |G_2|^s\big] .
\end{align*}
We now take the expectation over $V_\omega(u)$, which implies
\begin{align*}
   \E_u\Big[  \E_{v}\big[|G_1|^s\, |G_2|^s\big] \Big]\le     
    D_s \E_{v,\wt v}\Big[\, \E_u\big[|G_1|^s\, |G_2|^s\big] \, \Big].
\end{align*}
Applying \eqref{eqn:decoup-AW8.11} to $G_1$ and $G_2$ with the expectation over $V_\omega(u)$ gives
\begin{align*}
     \E_u\Big[  \E_{v}\big[|G_1|^s\, |G_2|^s\big] \Big]&\le     
    D_s^2 \E_{v,\wt v}\Big[\,  \E_{u}\big[|G_1|^s \big]\E_{u}\big[  |G_2|^s\big]  \, \Big]\\
    &=D_s^2 \E_{v,\wt v}\Big[\,  \E_{u}\big[|G_1|^s \big]\E_{\wt u}\big[  |G_2|^s\big]  \, \Big],
\end{align*}
where we have rewritten the second integral $\E_{u}\big[  |G_2|^s\big]=\E_{\wt u}\big[  |G_2|^s\big]$ in terms of an independent copy of $V_\omega(u)$.
Notice that after renaming the random variables, $G_1$ still depends on $V_\omega$ at $u,v$ and is independent of $\wt V_\omega$ at $ u, v$, while $G_2$ depends on $\wt V_\omega$ at $u,v$ and is independent of $V_\omega$ at $u,v$. Hence, the last expectation can be rewritten to give
\begin{align*}
   \E_{u,v}\big[|G_1|^s\, |G_2|^s\big]\le     
     D_s^2\,  \E_{u, v}\big[   |G_1|^s  \big]\, \E_{\wt u ,\wt v}\big[  |G_2|^s\big]= D_s^2\, \E_{u, v}\big[   |G_1|^s  \big]\, \E_{ u , v}\big[  |G_2|^s\big],
\end{align*}
where in the last equality we renamed the variables back to $u,v$.

Finally, we apply the a priori bound \eqref{eqn:priori-xy} to $\E_{u,v}[|G_1|^s]$ and then take the expectation over the remaining variables $\{V_\omega(x)\}_{x\neq u,v}$. This implies
\begin{align*}
   \E \big[|G_1|^s\, |G_2|^s\big] \big]=\E_{\neq u,v}\Big[\E_{u,v}\big[|G_1|^s\, |G_2|^s\big] \Big]\le C\,D_s^2\, \E \Big[  |G_2|^s\big]    \Big] , 
\end{align*}
which proves \eqref{eqn:GS-decoup}. 
 
\end{proof}

\noindent\textbf{Acknowledgments.}
\phantomsection
\addcontentsline{toc}{section}{Acknowledgments}
W. Wang is supported in part by the National Key R\&D Program of China (2024YFA1012302). S. Zhang is supported by the NSF grant DMS-2418611.


{
  \bigskip
  \vskip 0.08in \noindent --------------------------------------

\footnotesize
\medskip
L.~Shou, {Joint Quantum Institute, Department of Physics, University of Maryland, College Park, MD 20742, USA}\par\nopagebreak
    \textit{E-mail address}:  \href{mailto:lshou@umd.edu}{lshou@umd.edu}

\vskip 0.4cm

  W. ~Wang, {SKLMS, Academy of Mathematics and Systems Science, Chinese Academy of Sciences, Beijing 100190, China}\par\nopagebreak
  \textit{E-mail address}: \href{mailto:ww@lsec.cc.ac.cn}{ww@lsec.cc.ac.cn}
  
\vskip 0.4cm

S.~Zhang, {Department of Mathematics and Statistics, University of Massachusetts Lowell, 
Southwick Hall, 
11 University Ave.
Lowell, MA 01854
 }\par\nopagebreak
  \textit{E-mail address}: \href{mailto:shiwen\_zhang@uml.edu}{shiwen\_zhang@uml.edu}
}

\end{document}